\newif\ifregtemp
\regtemptrue
\ifregtemp
	\documentclass[10pt, a4paper, onecolumn, twoside]{amsart}

\usepackage{titlesec, color, amsthm}
\definecolor{blue1}{rgb}{0,0,0}
\renewcommand\thesection{\arabic{section}}

\titleformat{\section}[hang]{\color{blue1}\large\bfseries\sffamily}{\thesection}{0mm}{. }[]
\titleformat{\subsection}[hang] {\color{blue1}\bfseries\sffamily}{\thesubsection}{0em}{. }[]
\titleformat{\subsubsection}[hang] {\color{blue1}\sffamily}{\thesubsubsection}{0em}{. }[]
\titlespacing*{\section}{1em}{3.5ex plus .2ex minus .2ex}{1ex plus .2ex}
\titlespacing*{\subsection}{0em}{3ex plus .2ex minus .2ex}{1ex plus .2ex}
\titlespacing*{\subsubsection}{0em}{3ex plus .2ex minus .2ex}{1ex plus .2ex}

\renewenvironment{abstract}{{\color{blue1}\small\bfseries Abstract.}\footnotesize}{\par \vskip .1in}
\makeatletter
\def\@setauthors{
\begingroup 
\def \thanks{\protect\thanks@warning}
\trivlist \centering\footnotesize \@topsep30\p@\relax \advance\@topsep by -\baselineskip
\item\relax \author@andify \authors \def\\{\protect\linebreak} {\color{blue1}\large\authors} \endtrivlist \endgroup}
\def\@settitle{\centering{\color{blue1} \Large \bfseries \bfseries \@title \par}}
\makeatother
\usepackage[normal, small, up, textfont=it]{caption}
	\newtheorem{theorem}{Theorem}[section]
	\newtheorem{definition}[theorem]{Definition}
	\newtheorem{lemma}[theorem]{Lemma}
	\newtheorem{corollary}[theorem]{Corollary}

	\usepackage{cite}
\else
	\documentclass[nonumbib]{IMAIAI}
\fi
\usepackage{amsmath}
\usepackage{amssymb, bm, enumerate}
\usepackage[pdftex]{graphicx}
\usepackage{rotating}
\usepackage{epstopdf}
\usepackage{url}
\usepackage{color}

\newcommand{\Graph}{\ensuremath{\set{G}}}
\newcommand{\Group}{\ensuremath{\set{N}}}
\newcommand{\PSignals}{\ensuremath{\set{A}_\epsilon}}

\newcommand{\Lap}{\ensuremath{\ma{L}}}
\newcommand{\Fou}{\ensuremath{\ma{U}}}
\newcommand{\Eig}{\ensuremath{\ma{\Lambda}}}
\newcommand{\Meas}{\ensuremath{\ma{M}}}
\newcommand{\SelectGroup}{\ensuremath{\ma{N}}}
\newcommand{\Average}{\ensuremath{\ma{A}}}

\newcommand{\fou}{\ensuremath{\vec{u}}}

\newcommand{\eig}{\ensuremath{\lambda}}
\newcommand{\err}{\ensuremath{\vec{n}}}
\newcommand{\sig}{\ensuremath{\vec{x}}}
\newcommand{\meas}{\ensuremath{\vec{y}}}
\newcommand{\prob}{\ensuremath{\vec{p}}}
\newcommand{\average}{\ensuremath{\vec{a}}}

\newcommand{\nbVert}{\ensuremath{n}}
\newcommand{\nbVertRed}{\ensuremath{m}}
\newcommand{\nbClass}{\ensuremath{k}}
\newcommand{\nbGroups}{\ensuremath{N}}
\newcommand{\nbGroupsRed}{\ensuremath{s}}

\newcommand{\cumCoh}{\ensuremath{\nu}}
\newcommand{\reg}{\ensuremath{\gamma}}

\newcommand{\Rbb}{\ensuremath{\mathbb{R}}} 
\newcommand{\Nbb}{\ensuremath{\mathbb{N}}}
\newcommand{\Ebb}{\ensuremath{\mathbb{E}}}
\newcommand{\Pbb}{\ensuremath{\mathbb{P}}}

\renewcommand{\leq}{\ensuremath{\leqslant}}
\renewcommand{\geq}{\ensuremath{\geqslant}}

\newcommand{\adjoint}{\ensuremath{{\intercal}}}
\newcommand{\inv}[1]{\ensuremath{\frac{1}{#1}}}

\newcommand{\norm}[1]{\ensuremath{\left\| #1\right\|}}
\newcommand{\abs}[1]{\ensuremath{\left| #1 \right|}}

\newcommand{\ma}[1]{\ensuremath{\mathsf{#1}}}
\renewcommand{\vec}[1]{\ensuremath{\bm{#1}}}
\newcommand{\diag}{\ensuremath{{\rm diag}}}

\newcommand{\set}[1]{\ensuremath{\mathcal{#1}}}

\newcommand{\spann}{\ensuremath{{\rm span}}}

\newcommand{\ie}{\textit{i.e.}}
\newcommand{\eg}{\textit{e.g.}}
\renewcommand{\th}{\ensuremath{\text{th}}}
\newcommand{\ee}{\ensuremath{{\rm e}}}

\newcommand{\longtitle}{{Structured sampling and fast reconstruction of smooth graph signals}}
\newcommand{\GPlong}{{Gilles~Puy}}
\newcommand{\PPlong}{{Patrick P\'erez}}
\newcommand{\GPshort}{{G.~Puy}}
\newcommand{\PPshort}{{P.~P\'erez}}
\newcommand{\Technicolor}{{Technicolor, 975 Avenue des Champs Blancs, 35576 Cesson-S\'evign\'e, France.}}
\graphicspath{{figs/}}
\DeclareGraphicsExtensions{.pdf,.jpeg,.jpg,.eps,.png}
\ifregtemp
	\title[\longtitle]{\longtitle}
	\author[\GPshort]{\GPlong}
	\author[\PPshort]{\PPlong}
	\thanks{\GPshort\ and \PPshort\ are with \Technicolor.}
\else
	\title{\longtitle}
	\shorttitle{\longtitle}
	\shortauthorlist{\GPshort\ and \PPshort}
	\author{{
		\sc \GPlong}$^*$,\\[2pt]
		\Technicolor\\
		$^*${\email{Corresponding author: gilles.puy@technicolor.com}}\\[2pt]
		{\sc and}\\[6pt]
		{\sc \PPlong} \\[2pt]
		\Technicolor\\
		{patrick.perez@technicolor.com}
	}
\fi

\begin{document}

\maketitle

\begin{abstract}
{This work concerns sampling of smooth signals on arbitrary graphs. We first study a structured sampling strategy for such smooth graph signals that consists of a random selection of few pre-defined groups of nodes. The number of groups to sample to stably embed the set of $\nbClass$-bandlimited signals is driven by a quantity called the \emph{group} graph cumulative coherence. For some optimised sampling distributions, we show that sampling  $O(\nbClass\log(\nbClass))$ groups is always sufficient to stably embed the set of $\nbClass$-bandlimited signals but that this number can be smaller --~down to $O(\log(\nbClass))$~-- depending on the structure of the groups of nodes. Fast methods to approximate these sampling distributions are detailed. Second, we consider $\nbClass$-bandlimited signals that are nearly piecewise constant over pre-defined groups of nodes. We show that it is possible to speed up the reconstruction of such signals by reducing drastically the dimension of the vectors to reconstruct. When combined with the proposed structured sampling procedure, we prove that the method provides stable and accurate reconstruction of the original signal. Finally, we present numerical experiments that illustrate our theoretical results and, as an example, show how to combine these methods for interactive object segmentation in an image using superpixels.}
\ifregtemp
{}
\else
{Graph signal processing, compressive sampling, bandlimited graph signals.}
\fi
\end{abstract}
%


\section{Introduction}

This work was initially inspired by studying developments in edit propagation for interactive image or video manipulations, where the goal is to propagate operations made by a user in some parts of the image to the entire image, \eg, propagate foreground/background scribbles for object segmentation \cite{Boykov01,Grady06,chen2012manifold,chen2013image,xu2013sparse} or propagate manual color/tone modifications  \cite{chen2012manifold,chen2013image,levin2004colorization,lischinski2006interactive,xu2013sparse}.
First, a graph $\Graph$ modelling the similarities between the pixels of the image is built. Second, the user-edits specify values at some nodes (pixels) of $\Graph$. Finally, the complete signal is estimated by assuming that it is smooth on $\Graph$. The quality of the propagation depends on the structure of $\Graph$ and on the location of annotated nodes. If a part of the image is weakly connected to the rest, the user-edits do not propagate well to this region unless this region is edited directly. Therefore, highlighting beforehand which regions or groups of nodes are important to edit to ensure a good propagation would be a useful feature to facilitate user interactions. Furthermore, designing a fast reconstruction/propagation method is also important so that the user can visualise immediately the effect of his inputs. To this end, ``superpixels'' -- small groups of connected pixels where the image varies only little -- are sometimes used among other things to speed up computations, \eg, in \cite{chen2012manifold}. We address these problems from a graph signal processing point-of-view~\cite{shuman_SPMAG2013}. More precisely, we view them as a sampling problem where we should select a \emph{structured} set of nodes (regions to edit), ``measure'' the signal on this set (user-edits), and reconstruct the signal on the entire graph. We believe that the sampling strategy and the fast reconstruction method proposed here provide useful tools to optimize user's inputs and accelerate computations.

As mentionned above, the user-edits are propagated to the entire image by assuming that the global signal is smooth on $\Graph$. In the context of graph signal processing, smooth or $\nbClass$-bandlimited signal is a widely used and studied model. Several sampling methods have been designed to sample such signals. Pesenson introduced the notion of uniqueness set (of nodes) for $k$-bandlimited graph signals in~\cite{pesenson_paley, Pesenson:2011tu}. Two different $\nbClass$-bandlimited signals are also necessarily different when restricted to a uniqueness set. Therefore, one can sample all $\nbClass$-bandlimited signals on a uniqueness set. Then, Anis \emph{et al.}, \cite{anis_ICASSP2014,Anis_ARXIV2015} and Chen \emph{et al.}, \cite{chen_ICASSP2015,chen_TSP2015} proved that one can always find a uniquess set of $\nbClass$ nodes to sample all $\nbClass$-bandlimited signals. Finding this set is however computationally expensive. In~\cite{anis_arxiv2015_long}, graph spectral proxies are used to find such a set more efficiently. Yet the combinatorial problem that needs to be solved to find such a set makes the method still difficult to use for large graphs. Other authors used the idea of random sampling to be able to handle large graphs \cite{chen_TSP2015, chen_sampta2015, puy15b}. Recently, Puy \emph{et al.}~\cite{puy15b} proved that there always exists a random sampling strategy for which sampling $O(\nbClass\log(\nbClass))$ nodes is sufficient to stably embed the set of bandlimited signals. They also designed a fast and scalable algorithm to estimate the optimal sampling distribution. 

In this paper, we study first a random sampling strategy for $\nbClass$-bandlimited signals where we sample few groups of nodes instead of sampling individual nodes. We introduce the concept of \emph{local group graph coherence} that quantifies the importance of sampling each group. Second, in order to build a fast reconstruction technique for $\nbClass$-bandlimited signals, we use the intuition that a smooth graph signal is a signal that varies slowly from one node to its connected nodes. If we group few connected nodes together, we usually expect a bandlimited signal to be essentially constant on this set of nodes; as long as we do not group together too many weakly connected nodes. We propose here to use this property to accelerate the reconstruction of such $\nbClass$-bandlimited signals, \ie, $\nbClass$-bandlimited signals nearly piecewise constant over (pre-defined) groups of nodes. When combined with the proposed sampling technique, we prove that this fast method provides stable and accurate reconstructions of the signals of interest. Finally, we illustrate how to use these results for interactive object segmentation in an image, with  required node groups being superpixels.

\subsection{Contributions}

The random sampling strategy that we propose generalises the method proposed in~\cite{puy15b}. We use here a structured sampling strategy. Let us already acknowledge here that such strategies are also studied in the field of compressed sensing \cite{chauffert13, adcock14, boyer15, bigot16} and that some of our solutions are directly inspired by these works.

First, in this structured sampling setting, we show that the number of groups to sample is directly linked to a quantity called the \emph{group} graph cumulative coherence. This quantity generalises the concept of graph cumulative coherence introduced in~\cite{puy15b} and characterises how much the energy of $\nbClass$-bandlimited can stay concentrated in each group of nodes.

Second, we can then choose to sample the groups non-adaptively or to optimise the sampling distribution to minimise the number of groups to sample. With this optimised sampling distribution, our result shows that, in the worst case, sampling $O(\nbClass \log(\nbClass))$ groups of nodes is sufficient to ensure the reconstruction of all $\nbClass$-bandlimited signals. As each group can contain many nodes, we might have to sample a large number of nodes. This is the potential price to pay when sampling the nodes by groups. Fortunately, a smaller number of groups -- down to $O(\log(\nbClass))$ -- might already be sufficient if the groups are well designed.

Third, we describe a method to estimate the optimal sampling distribution without the need of computing the graph Fourier matrix and which is thus able to handle large graphs.

Fourth, estimating the optimal sampling distribution may still be too slow when a large number of groups are involved. We thus also present a sufficient recovery condition that involves a relaxed version of group graph cumulative coherence. The sampling distribution that minimises this relaxed coherence is fast to estimate for large graphs and large number of groups. With this sampling strategy, we prove that sampling $O(\nbClass \log(\nbClass))$ groups is always sufficient to ensure the reconstruction of all $\nbClass$-bandlimited signals. This strategy is mainly interesting at small $\nbClass$.

Finally, we propose a fast reconstruction method for $\nbClass$-bandlimited signals that are also nearly piecewise constant over pre-defined groups of nodes. We show that we can reduce drastically the dimension of the reconstruction problem for such signals. When the above sampling procedure is used to sample them, we prove that the proposed method provides accurate and stable recovery of this type of signals.

\subsection{Applications}
\label{sec:clustering}

The proposed sampling methods can have interest in several applications. For example, if one needs to develop sensors in a large scale network, it might be easier to deploy and install these sensors at nodes around few spatial locations instead of scattering them all over the network. Finding the best regions to monitor is thus important. In a social network, one might be interested in monitoring a signal defined over communities and thus should find which groups of users are the most important to sample. In semi-supervised learning, it may be easier to label jointly some similar nodes than individual nodes. Let us now detail such an example for semi-supervised classification.

The task we consider is interactive object segmentation in an image where, similarly to what is done in \cite{chen12} for instance, we build a graph $\Graph$ that models similarities between the pixels of the image, ask the user to label some regions depending on whether they are part of the object or not, diffuse the result on the complete image by supposing that indicator vectors of the object is smooth on $\Graph$. To propose the regions to label, we view this segmentation problem as a sampling problem of a smooth signal on $\Graph$ where the regions to label are chosen to ensure the recovery of a $\nbClass$-bandlimited signal. To choose the regions to label, we start by partitioning the image into superpixels, \eg, with SLIC technique~\cite{achanta12}. We see in Fig.~\ref{fig:image_example} that the superpixels divide the image into homogeneous regions. The superpixels follow the edges and most of them thus belong to either the tiger or the background but rarely both. One interest of dividing the image into superpixels is that it facilitates user interactions. It is easier to determine if a \emph{superpixel} belongs to the object of interest than if a \emph{pixel} belongs to the tiger, especially at the boundaries, as recently exploited for segmentation on touch-screens \cite{korinke2016superselect}. Using the sampling method that we developed we can propose to the user a small number of superpixels to label. Furthermore, this proposition is adapted to the structure of the graph. Another advantage of using superpixels is that the indicator function of the object, beyond being smooth on $\Graph$, is also approximately piecewise constant on the superpixels. We can thus use our reconstruction method to estimate rapidly the segmentation result from the labelled superpixels.

\begin{figure}[h]
\includegraphics[width=.48\linewidth]{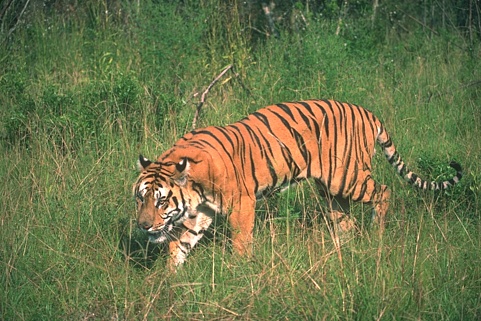}
\includegraphics[width=.48\linewidth]{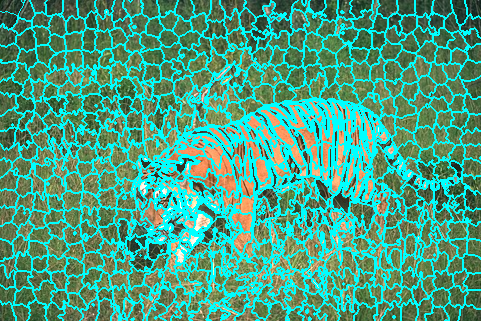}
\caption{\label{fig:image_example} One image (left) and its partition into superpixels with SLIC technique~\cite{achanta12} (right).}
\end{figure}
%

\subsection{Notations and definitions} 

We consider that $\Graph = \{\mathcal{V}, \mathcal{E}, \ma{W} \}$ is an undirected weighted graph, where $\mathcal{V}$ is the set of $\nbVert$ nodes, $\mathcal{E}$ is the set of edges, and $\ma{W} \in \Rbb^{\nbVert \times \nbVert}$ is the weighted adjacency matrix with nonnegative entries. We denote the graph Laplacian by $\Lap \in \Rbb^{\nbVert \times \nbVert}$. We assume that $\Lap$ is real, symmetric, and positive semi-definite, \eg, the combinatorial graph Laplacian $\Lap := \ma{D} - \ma{W}$, or the normalised one $\Lap := \ma{I} - \ma{D}^{-1/2} \ma{W} \ma{D}^{-1/2}$. The matrix $\ma{D} \in \Rbb^{\nbVert \times \nbVert}$ is the diagonal degree matrix and $\ma{I} \in \Rbb^{\nbVert \times \nbVert}$ is the identity matrix~\cite{chung_book1997}. The diagonal degree matrix $ \ma{D}$ has entries $D_{ii} := \sum_{i \neq j}\ma{W}_{ij}$.

We denote by $\Fou \in \Rbb^{\nbVert \times \nbVert}$ the orthonormal eigenvectors of $\Lap$ and by $0 = \eig_1 \leq \ldots \leq \eig_n$ the ordered real eigenvalues of $\Lap$. We have $\Lap = \Fou \Eig \Fou^\adjoint$, where $\Eig := \diag(\eig_1, \ldots, \eig_n) \in \Rbb^{\nbVert \times \nbVert}$. The matrix $\Fou$ is the graph Fourier basis~\cite{shuman_SPMAG2013}. For any signal $\sig \in \Rbb^{\nbVert}$ defined on the nodes of the graph $\Graph$, $\hat{\sig} = \Fou^\adjoint \sig$ contains the Fourier coefficients of $\sig$ ordered in increasing frequencies. This work deals with $k$-bandlimited signals $\sig \in \Rbb^{\nbVert}$ on $\Graph$, \ie, signals whose Fourier coefficients $\hat{\sig}_{k+1}, \ldots, \hat{\sig}_{\nbVert}$ are null. Let $\Fou_{\nbClass}$ be the restriction of $\Fou$ to its first $\nbClass$ vectors:
\begin{align}
\Fou_{\nbClass} := \left( \fou_1, \ldots, \fou_\nbClass \right) \in \Rbb^{\nbVert \times \nbClass}.
\end{align}
\begin{definition}[$k$-bandlimited signal on \Graph] A signal $\sig \in \Rbb^{\nbVert}$ defined on the nodes of the graph $\Graph$ is $k$-bandlimited with $\nbClass \in \Nbb\setminus\{0\}$ if $\sig \in \spann(\ma{U}_\nbClass)$, \ie, there exists $\vec{\eta} \in \Rbb^{k}$ such that
\begin{align}
\sig = \Fou_{\nbClass} \vec{\eta}.
\end{align}
\end{definition}
This definition was also used in~\cite{chen_TSP2015,anis_ICASSP2014, puy15b}. We assume that $\eig_{\nbClass} \neq  \eig_{\nbClass+1}$ to avoid any ambiguity in the definition of $\nbClass$-bandlimited signals.

Finally, for any matrix $\ma{X} \in \Rbb^{\nbVert_1 \times \nbVert_2}$, $\norm{\ma{X}}_2$ denotes its spectral norm and $\norm{\ma{X}}_F$ its Frobenius norm; when $n_1=n_2$, $\lambda_{\rm max}(\ma{X})$ denotes its largest eigenvalue and $\lambda_{\rm min}(\ma{X})$ its smallest eigenvalue. For any vector $\sig \in \Rbb^{\nbVert_1}$, $\norm{\sig}_2$ denotes the Euclidean norm of $\sig$. Depending on the context, $\sig_j$ may represent the $j^\th$ entry of the vector $\sig$ or the $j^\th$ column-vector of the matrix $\ma{X}$. The entry on the $i^\th$ row and $j^\th$ column of $\ma{X}$ is denoted by $\ma{X}_{ij}$. The identity matrix is denoted by $\ma{I}$ (its dimensions are determined by the context).

We present in Fig.~\ref{fig:variables} a representation of the important variables and processes involved in this paper in order to facilitate the understanding of the different results.

\begin{figure}[h]
\centering
\includegraphics[width=.90\linewidth]{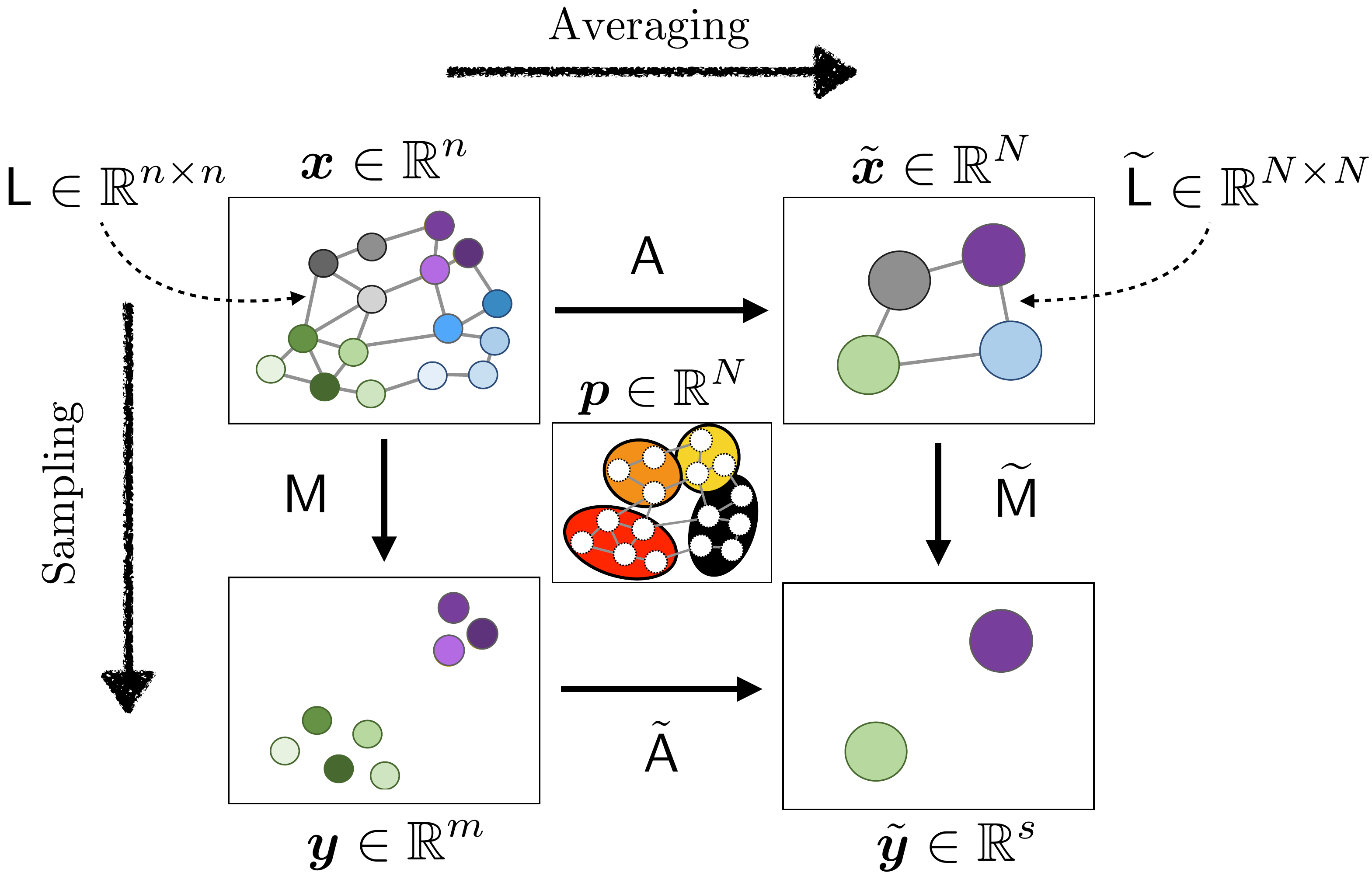}
\caption{\label{fig:variables} Summary of the main variables and processes involved in the presented sampling strategies. All variables are defined when needed hereafter.}
\end{figure}
%

\section{Sampling using groups of nodes}
\label{sec:rip}

In this section, we explain our sampling strategies, starting with the definition of the groups of nodes.

\subsection{Grouping the nodes}

We consider that the $\nbVert$ nodes of $\Graph$ are divided into $\nbGroups$ different groups $\Group_1, \ldots, \Group_\nbGroups \subseteq \{1, \ldots, \nbVert\}$. The size of the $\ell^\th$ group is denoted $\abs{\Group_{\ell}}$. We suppose that these groups form a partition of $\{1, \ldots, \nbVert\}$, so that each node belongs to exactly one group. We have
\begin{align}
\cup_{\ell=1}^N \, \Group_{\ell} \; & = \; \{1, \ldots, \nbVert \},
\text{ and } \Group_{\ell} \, \cap \, \Group_{\ell'} = \emptyset.
\end{align}
For the object segmentation application discussed in the introduction, these groups represent the superpixels. However, we do not impose the groups to be made of neighbouring nodes in the graph; they can be made of nodes ``far'' from each other.

For each group $\Group_{\ell} = \{\nbVert_1^{(\ell)}, \ldots, \nbVert_{\abs{\Group_{\ell}}}^{(\ell)} \}$, we associate a matrix $\SelectGroup^{(\ell)}~\in~\Rbb^{\abs{\Group_{\ell}} \times \nbVert}$ that restricts a graph signal to the nodes appearing in $\Group_{\ell}$, \ie,
\begin{align}
\SelectGroup^{(\ell)}_{ij} := 
\left\{
\begin{array}{ll}
1 & \text{for } j = \nbVert_i^{(\ell)},\\
0 & \text{otherwise}.
\end{array}
\right.
\end{align}
Note that
\begin{align}
\label{eq:group_tight_frame}
\sum_{\ell=1}^\nbGroups {\SelectGroup^{(\ell)}}^\adjoint \SelectGroup^{(\ell)} = \ma{I}.
\end{align}

The case of overlapping groups can be handled by changing the definition of $\SelectGroup^{(\ell)}$ to
\begin{align}
\SelectGroup^{(\ell)}_{ij} := 
\left\{
\begin{array}{ll}
\beta_{\nbVert_i^{(\ell)}}^{-1/2} & \text{for } j = \nbVert_i^{(\ell)},\\
0 & \text{otherwise},
\end{array}
\right.
\end{align}
where $1 \leq \beta_{i} \leq \nbGroups$, $i = 1, \ldots, \nbVert$, is the number of times node $i$ appears in the different groups $\Group_1, \ldots, \Group_\nbGroups$. Equation \eqref{eq:group_tight_frame} also holds in this case. All results presented in Section~\ref{sec:rip} are valid for overlapping groups with this definition of $\SelectGroup^{(\ell)}$.

\subsection{Sampling the groups}

The sampling procedure consists in selecting $\nbGroupsRed$ groups out of the $\nbGroups$ available ones. In the application of Section~\ref{sec:clustering}, it corresponds to the selection of the superpixels to label. We select these groups at random using a sampling distribution on $\{ 1, \ldots, \nbGroups \}$ represented by a vector $\prob \in \Rbb^\nbGroups$. The probability of selecting the $\ell^\th$ group is $\prob_{\ell}$. We assume that $\prob_{\ell} > 0$ for all $\ell = 1, \ldots, \nbGroups$, so that all groups may be selected with a non-zero probability. We obviously have $\sum_{\ell=1}^\nbGroups \prob_{\ell} = 1$.

The indices $\Omega := \{ \omega_1, \ldots, \omega_{\nbGroupsRed} \}$ of the selected groups are obtained by drawing independently --~thus with replacements~-- $\nbGroupsRed$ indices from the set $\{1, \ldots, \nbGroups \}$ according to $\prob$, \ie,
\begin{align}
\Pbb(\omega_j = \ell) = \prob_{\ell},\quad \forall j \in \{1, \ldots, \nbGroupsRed \} \text{ and } \forall \ell \in \{1, \ldots, \nbGroups \}.
\end{align}
The selected groups are $\Group_{\omega_1}, \ldots, \Group_{\omega_\nbGroupsRed}$ and the total number of selected nodes is
\begin{align}
\nbVertRed := \sum_{j=1}^{\nbGroupsRed} \abs{\Group_{\omega_j}}.
\end{align}
Once the groups are selected, we build the sampling matrix $\Meas \in \Rbb^{\nbVertRed \times \nbVert}$ that satisfies
\begin{align}
\label{eq:subsampling_matrix_def}
\Meas := 
\left(
\begin{array}{c}
\SelectGroup^{(\omega_1)}\\
\vdots\\
\SelectGroup^{(\omega_\nbGroupsRed)}
\end{array}
\right),
\end{align}
and which restricts any signal to the nodes belonging to the selected groups. For a signal $\sig \in \Rbb^n$ defined on the nodes of $\Graph$, its sampled version $\meas \in \Rbb^\nbVertRed$ satisfies
\begin{align}
\meas := \Meas \sig.
\end{align}

Our goal now is to determine what number $\nbGroupsRed$ is enough to ensure that all $\nbClass$-bandlimited signals can be reconstructed from its sampled version obtained with $\Meas$. To conduct this study, we need to define few more matrices. First, we associate the matrix
\begin{align}
\ma{P}^{(\ell)} := \prob_{\ell}^{-1/2} \; \ma{I} \in \Rbb^{\abs{\Group_{\ell}} \times \abs{\Group_{\ell}}},
\end{align}
to each group $\Group_{\ell}$. Then, once the groups are drawn, we construct the block diagonal matrix $\ma{P} \in \Rbb^{\nbVertRed \times \nbVertRed}$
\begin{align}
\label{eq:subsampled_weight_matrix}
\ma{P} := \diag\left( \ma{P}^{(\omega_1)}, \ldots, \ma{P}^{(\omega_\nbGroupsRed)} \right).
\end{align}
This matrix takes into account the probability of sampling each group and will be used to rescale $\meas$ for norm preservation. This matrix ensures that $\nbGroupsRed^{-1} \,\Ebb_{\Omega} \norm{\ma{P} \meas}_2^2 = \nbGroupsRed^{-1} \, \Ebb_{\Omega} \norm{\ma{P} \Meas \sig}_2^2 = \norm{\sig}_2^2$.\footnote{This property is a consequence of \eqref{eq:expected_value_proof}, proved in Appendix~\ref{app:proof_rip}.} Both matrices $\ma{P}$ and $\ma{M}$ depend on $\Omega$ and are random.

\subsection{Group graph coherence}

The sampling procedure in \cite{puy15b} is similar to the ones proposed here at the difference that the nodes are sampled individually and not by groups. It was proved there that the number of nodes to sample is driven by a quantity called the graph coherence. This quantity measures how the energy of $\nbClass$-bandlimited signals spreads over the \emph{nodes}. Similarly, we prove here that the number of groups to sample is driven by a quantity that measures the energy of $\nbClass$-bandlimited signals spreads over the \emph{groups}. We now introduce this quantity.

The matrix $\SelectGroup^{(\ell)} \Fou_\nbClass$ is the matrix that restricts a $\nbClass$-bandlimited signal to the nodes belonging to $\Group_{\ell}$. Therefore, 
\begin{align}
\norm{\SelectGroup^{(\ell)} \Fou_\nbClass}_2 = \sup_{\vec{\eta} \in \Rbb^k : \norm{\vec{\eta}}_2=1} \norm{\SelectGroup^{(\ell)} \Fou_\nbClass \vec{\eta}}_2
\end{align}
measures the energy on the nodes $\Group_{\ell}$ of the normalised $\nbClass$-bandlimited signal that is most concentrated on $\Group_{\ell}$. This energy varies between $0$ and $1$. When this energy is close to $1$, there exists a $\nbClass$-bandlimited signal whose energy is essential concentrated on $\Group_{\ell}$. This signal lives only on the nodes in $\Group_{\ell}$ and does not spread elsewhere. On the contrary, when this energy is close to $0$, there is no $\nbClass$-bandlimited signal living only on $\Group_{\ell}$.

The sampling distribution $\prob$ is adapted to the graph and the structure of the groups if: whenever $\|\SelectGroup^{(\ell)} \Fou_\nbClass \|_2$ is high, $\prob_{\ell}$ is high; whenever $\|\SelectGroup^{(\ell)} \Fou_\nbClass \|_2$ is small, $\prob_{\ell}$ is small. In other words, the ratio between $\|\SelectGroup^{(\ell)} \Fou_\nbClass \|_2$ and $\prob_{\ell}$ should be as constant as possible. This ensures that the groups where some $\nbClass$-bandlimited signals are concentrated are sampled with higher probability. Similarly to what was done in~\cite{puy15b} with individual nodes, we define the group graph weighted coherence as the largest ratio between $\|\SelectGroup^{(\ell)} \Fou_\nbClass \|_2$ and $\prob_{\ell}^{-1/2}$.
\begin{definition}[Group graph cumulative coherence]
The group graph cumulative coherence of order $\nbClass$ is
\begin{align}
\cumCoh_{\prob} := \max_{1 \leq \ell \leq \nbGroups} \left\{ \prob_{\ell}^{-1/2} \norm{\SelectGroup^{(\ell)} \Fou_\nbClass}_2 \right\}.
\end{align}
The quantity $\|\SelectGroup^{(\ell)} \Fou_\nbClass \|_2$ is called the local group graph coherence.
\end{definition}
In the extreme case where the groups $\Group_1, \ldots, \Group_\nbGroups$ reduce all to singletons, we recover the definition of the graph weighted coherence introduced in~\cite{puy15b}. 

It is easy to prove that $\cumCoh_{\prob}$ is lower bounded by $1$. Indeed, for any $\vec{\eta} \in \Rbb^\nbClass$ with $\norm{\vec{\eta}}_2 = 1$, we have
\begin{align}
1 
& = \norm{\Fou_k \vec{\eta}}_2^2 
= \sum_{\ell=1}^{N} \norm{\SelectGroup^{(\ell)} \Fou_k \vec{\eta}}_2^2
= \sum_{\ell=1}^{N} \prob_{\ell} \cdot \frac{\norm{\SelectGroup^{(\ell)}\Fou_k \vec{\eta}}_2^2}{\prob_{\ell}}
\leq \norm{\prob}_1 \cdot \max_{1 \leq \ell \leq \nbGroups} \left\{\frac{\norm{\SelectGroup^{(\ell)}\Fou_k \vec{\eta}}_2^2}{\prob_{\ell}} \right\} \nonumber \\
& = \max_{1 \leq \ell \leq \nbGroups} \left\{\frac{\norm{\SelectGroup^{(\ell)}\Fou_k \vec{\eta}}_2^2}{\prob_{\ell}} \right\}
\leq \max_{1 \leq \ell \leq \nbGroups} \left\{\frac{\norm{\SelectGroup^{(\ell)}\Fou_k}_2^2}{\prob_{\ell}} \right\}
= \cumCoh_{\prob}^2.
\end{align}
We have $\cumCoh_{\prob} = 1$ in, for example, the degenerated case where $\Group_1 = \{1, \ldots, \nbVert\}$.

\subsection{Stable embedding}

We now have all the tools to present our main theorem that shows that sampling $O(\cumCoh_{\prob}^2 \log(\nbClass))$ groups is sufficient to stably embed the whole set of $\nbClass$-bandlimited signals. Hence, it is possible to reconstruct any $\sig \in \spann(\Fou_k)$ from its measurements $\meas = \Meas \sig$.

\begin{theorem}[Restricted isometry property - RIP]
\label{th:rip}
Let $\Meas$ be a random subsampling matrix constructed as in \eqref{eq:subsampling_matrix_def} using the groups $\Group_1, \ldots, \Group_\nbGroups$ and the sampling distribution $\prob$. For any $\delta, \xi \in (0, 1)$, with probability at least $1-\xi$,
\begin{align}
\label{eq:RIP}
(1 - \delta) \norm{\sig}_2^2 \leq \inv{\nbGroupsRed} \norm{\ma{P} \Meas \; \sig}_2^2 \leq (1 + \delta) \norm{\sig}_2^2
\end{align}
for all $\sig \in \spann(\Fou_{\nbClass})$ provided that
\begin{align}
\label{eq:sampling_condition}
\nbGroupsRed \geq \frac{3}{\delta^2} \; \cumCoh_{\prob}^2 \; \log\left( \frac{2\nbClass}{\xi} \right).
\end{align}
\end{theorem}
\begin{proof}
See Appendix~\ref{app:proof_rip}.
\end{proof}

In the above theorem, we recall that $\nbGroupsRed$ is the number of selected groups, each of them containing several nodes. We thus control the number of groups to sample and not directly the number of nodes. As the lower bound on $\cumCoh_{\prob}$ is $1$, sampling $O(\log(\nbClass))$ groups might already be sufficient if the groups and the sampling distribution are well-designed.

The number of groups to sample is driven by $\cumCoh_{\prob}$, which itself depends on the structure of the groups $\Group_1, \ldots, \Group_\nbGroups$ and on the sampling distribution $\prob$. To reduce the number of samples to measure, we might optimise the structure of the groups and the sampling distribution. For example, if we were able to construct $\nbVert/L$ groups ($L\geq1$) such that $\|\SelectGroup^{(\ell)} \Fou_k\|_2 \approx (L/\nbVert)^{1/2}$ -- \ie, no $\nbClass$-bandlimited signals have more than $100 \cdot (L/\nbVert)^{1/2}$ percent of its energy concentrated in each group -- then setting $\prob_{\ell} = L/\nbVert$, $\ell = 1, \ldots, \nbVert/L$, would yield $\cumCoh \approx 1$. In this case, sampling one group is enough to embed the set of $\nbClass$-bandlimited signals. However, it is not obvious how we can construct such groups in practice and we might not even have the flexibility to modify the structure of the groups. In such a case, the only possibility to reduce the number of measurements is to optimise the sampling distribution $\prob$ to minimise $\cumCoh_{\prob}$.

The sampling distribution minimizing the coherence $\cumCoh_{\prob}$ is the distribution $\prob^* \in \Rbb^\nbGroups$ that satisfies
\begin{align}
\prob^*_{\ell} := \frac{\norm{\SelectGroup^{(\ell)}\Fou_k}_2^2}{\sum_{\ell'=1}^\nbGroups \norm{\SelectGroup^{(\ell')}\Fou_k}_2^2}, \text{ for all } \ell \in \{1, \ldots, \nbGroups \},
\label{eq:pstar}
\end{align}
and for which
\begin{align}
\cumCoh_{\prob^*}^2 = \sum_{\ell=1}^\nbGroups \norm{\SelectGroup^{(\ell)}\Fou_k}_2^2.
\end{align}
Indeed, let $\vec{p}' \neq \prob^*$ be another sampling distribution. As the entries of $\vec{p}'$ and $\prob^*$ are nonnegative and sum to $1$, we necessarily have $\vec{p}'_{\ell'} < \prob_{\ell'}^*$ for some $\ell'>0$. Then, 
\begin{align}
\cumCoh_{\vec{p}'}^2
\geq {\vec{p}'_{\ell'}}^{-1} \norm{\SelectGroup^{(\ell')}\Fou_k}_2^2
> 
{\prob_{\ell'}^*}^{-1} \norm{\SelectGroup^{(\ell')}\Fou_k}_2^2
=
\cumCoh_{\prob^*}^2,
\end{align}
where the last equality is obtained by replacing $\prob_{\ell'}^*$ with its value. Therefore, $\cumCoh_{\vec{p}'} > \cumCoh_{\prob^*}$ for any $\vec{p}' \neq \prob^*$. As similar proof can be found in, \eg, \cite{chauffert13} where the authors derive the optimal sampling distribution for a compressive system.

We notice that
\begin{align}
\cumCoh_{\prob^*}^2 = \sum_{\ell=1}^\nbGroups \norm{\SelectGroup^{(\ell)}\Fou_k}_2^2 \leq \sum_{\ell=1}^\nbGroups \norm{\SelectGroup^{(\ell)}\Fou_k}_F^2 = \nbClass.
\end{align}
Hence, by using this distribution, \eqref{eq:sampling_condition} shows that sampling $O(\nbClass\log(\nbClass))$ groups is always sufficient to sample all $\nbClass$-bandlimited signals. The exact number is proportional to $\cumCoh_{\prob^*}^2 \log(\nbClass)$. This is not in contradiction with the fact that at least $\nbClass$ measurements are required in total as one group contains at least one node. We also have
\begin{align}
\cumCoh_{\prob^*}^2 = \sum_{\ell=1}^\nbGroups \norm{\SelectGroup^{(\ell)}\Fou_k}_2^2 \leq \nbGroups,
\end{align}
as $\norm{\SelectGroup^{(\ell)}\Fou_k}_2^2 \leq 1$. Therefore, in any case, the bound never suggests to sample much more than $\nbGroups$ groups, as one would expect it.

We recall that the results in~\cite{puy15b} proves that it is always sufficient to sample $O(\nbClass\log(\nbClass))$ nodes to embed the set of $\nbClass$-bandlimited signals. When sampling the nodes by groups, it is the number of groups to sample that should be $O(\nbClass\log(\nbClass))$. This can be a large number of individual nodes but this is the potential price to pay when sampling the nodes by groups.

Variable density sampling \cite{puy11, krahmer12, adcock14} and structured sampling \cite{chauffert13, boyer15, bigot16} are also important topics in compressed sensing. The method proposed here is closely inspired by these studies, especially by \cite{boyer15, bigot16}. Our results thus share several similarities with these works. We however benefit from a simpler signal model and take advantage of the graph structure to refine the results, propose simpler decoders to reconstruct the signal, and design efficient algorithms to estimate $\prob^*$.

\subsection{A more practical result at small $k$'s}

Optimising the sampling distribution reduces to estimating the spectral norm of the matrices $\SelectGroup^{(\ell)} \Fou_\nbClass$. We will present in Section~\ref{sec:estimate_sampling} a method avoiding the computation of $\Fou_\nbClass$. However, the method might still be too slow when a large number of groups is involved as it requires an estimation of a spectral norm for each group separately. It is thus interesting to characterise the performance of the proposed method using other quantities easier to compute. 

In this section, we present results involving the following quantity
\begin{align}
\bar{\cumCoh}_{\prob} := \max_{1 \leq \ell \leq \nbGroups} \left\{ \prob_{\ell}^{-1/2} \norm{\SelectGroup^{(\ell)} \Fou_\nbClass}_F \right\}.
\end{align}
The only difference between $\cumCoh_{\prob}$ and $\bar{\cumCoh}_{\prob}$ is that we substituted the Frobenius norm for the spectral norm. As $\norm{\ma{X}}_2 \leq \norm{\ma{X}}_F$ for any matrix $\ma{X}$, we have $\cumCoh_{\prob} \leq \bar{\cumCoh}_{\prob}$; hence the results involving $\bar{\cumCoh}_{\prob}$ will be more pessimistic than those involving $\cumCoh_{\prob}$. We have $\bar{\cumCoh}_{\prob} \geq \nbClass$. Indeed,
\begin{align}
k 
& = \norm{\Fou_k}_F^2 
= \sum_{\ell=1}^{N} \norm{\SelectGroup^{(\ell)} \Fou_k}_F^2
= \sum_{\ell=1}^{N} \prob_{\ell} \cdot \frac{\norm{\SelectGroup^{(\ell)}\Fou_k}_F^2}{\prob_{\ell}}
\leq \norm{\prob}_1 \cdot \max_{1 \leq \ell \leq \nbGroupsRed} \left\{\frac{\norm{\SelectGroup^{(\ell)}\Fou_k}_F^2}{\prob_{\ell}} \right\}
= \bar{\cumCoh}_{\prob}^2.
\end{align}
As with $\cumCoh_{\prob}$, the lower bound is also attained in, for example, the degenerated case where $\Group_1 = \{1, \ldots, \nbVert\}$. As $\cumCoh_{\prob} \leq \bar{\cumCoh}_{\prob}$, we have the following corollary to Theorem~\ref{th:rip}.
\begin{corollary}
Let $\Meas$ be a random subsampling matrix constructed as in \eqref{eq:subsampling_matrix_def} using the groups $\Group_1, \ldots, \Group_\nbGroups$ and the sampling distribution $\prob$. For any $\delta, \xi \in (0, 1)$, with probability at least $1-\xi$, \eqref{eq:RIP} holds for all $\sig \in \spann(\Fou_{\nbClass})$ provided that
\begin{align}
\label{eq:sampling_condition_pes}
\nbGroupsRed \geq \frac{3}{\delta^2} \; \bar{\cumCoh}_{\prob}^2 \; \log\left( \frac{2\nbClass}{\xi} \right).
\end{align}
\end{corollary}
\begin{proof}
As $\bar{\cumCoh}_{\prob} \geq \cumCoh_{\prob}$, \eqref{eq:sampling_condition_pes} implies \eqref{eq:sampling_condition}. Theorem~\ref{th:rip} then proves that \eqref{eq:RIP} holds with probability at least $1-\xi$ for all $\sig \in \spann(\Fou_{\nbClass})$. 
\end{proof}

The sufficient condition~\eqref{eq:sampling_condition_pes} can be much more pessimistic than~\eqref{eq:sampling_condition}. Indeed, as we have $\bar{\cumCoh}_{\prob}^2 \geq \nbClass$, Condition~\eqref{eq:sampling_condition_pes} suggests to always sample more than $O(\nbClass\log(\nbClass))$ groups, while we know that sampling $O(\log(\nbClass))$ can be enough. The interest of this result is thus in the regime where $\nbClass$ is small.

As we have done it with $\cumCoh_{\prob}$, we can also optimise $\prob$ to minimise $\bar{\cumCoh}_{\prob}$. The sampling distribution that minimises $\bar{\cumCoh}_{\prob}$ is the distribution $\vec{q}^*$ satisfying
\begin{align}
\vec{q}^*_{\ell} := \frac{\norm{\SelectGroup^{(\ell)}\Fou_k}_F^2}{\nbClass}, \text{ for all } l \in \{1, \ldots, \nbGroups \},
\label{eq:qstar}
\end{align}
and for which
\begin{align}
\bar{\cumCoh}_{\vec{q}^*}^2 = \nbClass.
\end{align}
With this distribution, \eqref{eq:sampling_condition_pes} proves that sampling $O(\nbClass\log(\nbClass))$ groups is always enough to sample all $\nbClass$-bandlimited signals. This result is particularly interesting at small $\nbClass$'s because estimating $\vec{q}^*$ is much easier and faster than estimating $\prob^*$ (see Section~\ref{sec:estimate_sampling}).

In some particular cases, this result can be interesting  also at large $\nbClass$'s. Indeed, \eqref{eq:sampling_condition_pes} is a pessimistic bound. In reality, we may have $\cumCoh_{\vec{q}^*}^2 \leq \nbClass$ so that, according to \eqref{eq:sampling_condition}, fewer samples than $O(\nbClass\log(\nbClass))$ are actually sufficient when using $\vec{q}^*$. Furthermore, $\vec{q}^*$ might actually be close to $\prob^*$, in which case we would reach almost optimal results with $\vec{q}^*$. The occurrence of these events is however quite difficult to predict and depends on the structure of the groups and of the graph.

Note that we need the knowledge of the truncated Fourier matrix $\Fou_k$ to compute the distributions $\vec{p}^*$ and $\vec{q}^*$. Computing $\Fou_k$ can be intractable for large graphs. We will present a method to overcome this issue in Section~\ref{sec:estimate_sampling}. We continue first by explaining how to estimate $\sig$ from its measurements.

%
\section{Fast reconstruction}
\label{sec:fast_decoder}

Once the signal has been sampled, we shall also be able to reconstruct it. In \cite{puy15b}, the authors propose to estimate the original signal by solving
\begin{align}
\label{eq:practical_decoder}
\min_{\vec{z} \in \Rbb^\nbVert} \norm{\ma{P}(\Meas \vec{z} - \meas)}_2^2 + \reg \; \vec{z}^\adjoint g(\Lap) \vec{z},
\end{align}
where $\reg>0$ and $g$ is a nonnegative nondecreasing polynomial. We have
\begin{align}
g(\Lap) := \sum_{i=0}^d \alpha_i \, \Lap^i = \sum_{i=0}^d \alpha_i \, \Fou \Eig^i \Fou^\adjoint,
\end{align}
where $\alpha_0, \ldots, \alpha_{d} \in \Rbb$ are the coefficients of the polynomial and $d \in \Nbb$ its degree. These polynomial's parameters can be tuned to improve the quality of the reconstruction. The role of $g$ can be viewed as a filter on $\Graph$ and should ideally be a high-pass filter. The matrix $\ma{P}$ is introduced to account for the RIP. The advantage of this method is that it can be solved efficiently even for large graph for example by conjugate gradient. Indeed, each step of this algorithm can be implemented by using only matrix-vector multiplications with $\ma{P}\Meas$ and $\Lap$, which are both sparse matrices. The matrix $g(\Lap)$ does not have to be computed explicitly.

For brevity, we do not recall the theorem proving that~\eqref{eq:practical_decoder} provides accurate and stable recovery of all $\nbClass$-bandlimited signals. However, this theorem applies as soon as the restricted isometry property~\eqref{eq:RIP} holds, and thus applies here when~\eqref{eq:sampling_condition} holds. The reconstruction quality is controlled by $g(\eig_\nbClass)$ and the ratio ${g(\eig_\nbClass)}/{g(\eig_{\nbClass+1})}$. One should seek to design a filter $g$ such that these quantities are as close as possible to zero to improve the reconstruction quality.

We propose now a method to obtain a faster estimation of the original signal when it is nearly piecewise constant.

\subsection{Piecewise constant graph signals}

Before continuing, we want to stress that we consider non-overlapping groups $\Group_1, \ldots, \Group_\nbGroups$ in this rest of Section~\ref{sec:fast_decoder}.

If a graph signal is nearly piecewise constant over the groups $\Group_1, \ldots, \Group_{\nbGroups}$ then reconstructing the mean values of this signal for each group is enough to obtain a good approximation of the original signal. Instead of estimating $\nbVert$ unknowns, we reduce the estimation to $\nbGroups$ unknowns. When $\nbGroups \ll \nbVert$, this is a large reduction of dimension yielding a significant speed up and memory reduction. 

The fact that a signal $\sig$ is piecewise constant over the groups $\Group_1, \ldots, \Group_\nbGroups$ is characterized as follows. We construct the averaging row-vectors $\average^{(\ell)} \in \Rbb^{1 \times \nbVert}$ that satisfy
\begin{align}
\average^{(\ell)} := \frac{\vec{1}^\adjoint \SelectGroup^{(\ell)}}{\abs{\Group_{\ell}}^{1/2}},
\end{align}
and the matrix
\begin{align}
\Average := 
\left(
\begin{array}{c}
\average^{(1)}\\
\vdots\\
\average^{(\nbGroups)}
\end{array}
\right) \in \Rbb^{\nbGroups \times \nbVert}.
\end{align}
As the groups do not overlap, we have 
\begin{align}
\label{eq:tight_frame_average}
\Average\Average^\adjoint = \ma{I},
\end{align}
hence $\norm{\Average}_2 = 1$. Applying $\Average$ to $\sig$ provides $\nbGroups$ values, each one of them corresponding to the sum of the values of $\sig$ within the group $\Group_{\ell}$, scaled by $\abs{\Group_{\ell}}^{-1/2}$. Then, applying $\Average^\adjoint$ to $\Average\sig$ gives an approximation of $\sig$ where the values in the vector $\Average^\adjoint\Average\sig$ are constant within each group; this is a piecewise constant vector over the groups. The value of $\Average^\adjoint\Average\sig$ appearing within the group $\Group_{\ell}$ is exactly the average of $\sig$ within $\Group_{\ell}$. Saying that $\sig$ is nearly constant within each group corresponds to assume that 
\begin{align}
\label{eq:blockconstant_signal}
\norm{\Average^\adjoint \Average \sig - \sig}_2 \leq \epsilon \norm{\sig}_2,
\end{align}
where $\epsilon \geq 0$ is a small value. The signal model of interest in this section is thus
\begin{align}
\PSignals := \left\{ \sig \in \spann(\Fou_\nbClass) \;\; \vert \;\; \norm{(\Average^\adjoint \Average - \ma{I}) \; \sig}_2 \leq \epsilon \norm{\sig}_2 \right\}.
\end{align}
%

\subsection{Reducing the dimension}

To build a fast algorithm exploiting the above property, we use a reconstruction method similar to~\eqref{eq:practical_decoder} but involving vectors of smaller dimension. We define the averaged vector $\tilde{\sig} := \Average \sig \in \Rbb^{\nbGroups}$ of dimension $\nbGroups$. As $\sig \in \PSignals$, estimating $\tilde{\sig}$ is enough to get a good approximation of $\sig$ -- we just need to multiply it with $\Average^\adjoint$. Furthermore, as $\sig$ is nearly piecewise constant over the groups $\Group_1, \ldots, \Group_\nbGroups$, by construction of the matrix $\Meas$, the measurement vector $\meas = \Meas \sig$ is also almost piecewise constant over the sampled groups $\Group_{\omega_1}, \ldots, \Group_{\omega_\nbGroupsRed}$. We thus average $\meas$ over these groups by multiplying it with the matrix $\widetilde{\Average} \in \Rbb^{\nbGroupsRed \times \nbVertRed}$ that satisfies
\begin{align}
\widetilde{\Average}_{ji} := 
\left\{
\begin{array}{cc}
\abs{\Group_{\omega_{j}}}^{-1/2}
& \text{for } \sum_{j'=1}^{j-1} \abs{\Group_{\omega_{j'}}} \leq \; i  \;\leq \sum_{j'=1}^{\ell} \abs{\Group_{\omega_{j'}}}, \\
0 & \text{otherwise}.
\end{array}
\right.
\end{align}
We obtain
\begin{align}
\tilde{\meas} 
:= \widetilde{\Average} \meas 
= \widetilde{\Average} \Meas \sig \in \Rbb^{\nbGroupsRed}.
\label{eq:ytilde1}
\end{align}
We now have to link $\tilde{\meas}$ to $\tilde{\sig}$. We create the matrix $\widetilde{\Meas} \in \Rbb^{\nbGroupsRed \times \nbGroups}$ that restricts the $\nbGroups$ mean value of $\tilde{\sig}$ to the $\nbGroupsRed$ mean value of the selected groups, \ie,
\begin{align}
\label{eq:small_subsampling_matrix}
\widetilde{\Meas}_{j i} := 
\left\{
\begin{array}{cc}
1 & \text{if } i = \omega_{j},\\
0 & \text{otherwise}.
\end{array}
\right.
\end{align}
We have therefore
\begin{align}
\label{eq:ytilde2}
\tilde{\meas} = \widetilde{\Meas} \, \tilde{\sig}.
\end{align}

The goal is now to estimate $\tilde{\sig}$ from $\tilde{\meas}$. To ensure that the reconstruction method is stable to measurement noise, we do not consider the perfect scenario above but instead the scenario where
\begin{align}
\tilde{\meas} = \widetilde{\Meas}\tilde{\sig} + \tilde{\err},
\end{align}
and $\tilde{\err} \in \Rbb^\nbGroupsRed$ models noise. We now need a regularisation term to estimate $\tilde{\sig}$. We obtain this term by reducing the dimension of the regularisation involving the Laplacian $\Lap$ in~\eqref{eq:practical_decoder}. We compute
\begin{align}
\widetilde{\Lap} := \Average \, g(\Lap) \, \Average^\adjoint = (\Average \Fou) \, g(\Eig) \, (\Average \Fou)^\adjoint \in \Rbb^{\nbGroups \times \nbGroups}.
\end{align}
Note that $\widetilde{\Lap}$ is a symmetric positive definite matrix. Like $g(\Lap)$, it can thus be used as a regularisation. We thus propose to estimate $\tilde{\sig}$ by solving
\begin{align}
\label{eq:fast_decoder}
\min_{\tilde{\vec{z}} \in \Rbb^{\nbGroups}} \norm{\widetilde{\ma{P}} (\widetilde{\Meas} \tilde{\vec{z}} - \tilde{\meas})}_2^2 + \reg \; \tilde{\vec{z}}^\adjoint \, \widetilde{\Lap} \, \tilde{\vec{z}},
\end{align}
where $\reg>0$ and $\widetilde{\ma{P}} \in \Rbb^{\nbGroupsRed \times \nbGroupsRed}$ is the diagonal matrix with entries satisfying
\begin{align}
\label{eq:small_weight_matrix}
\widetilde{\ma{P}}_{jj} := \prob_{\omega_{j}}^{-1/2}.
\end{align}
Let $\tilde{\sig}^* \in \Rbb^{\nbGroups}$ be a solution of~\eqref{eq:fast_decoder}. We finally obtain an estimation of $\sig$ by computing $\Average^\adjoint \tilde{\sig}^*$.

In the particular case where $g(\cdot)$ is the identity, one can notice that 
\begin{align}
\widetilde{\Lap}_{\ell\ell'} = \sum_{(i,j) \in \Group_\ell \times \Group_{\ell'}} \frac{\Lap_{ij}}{\abs{\Group_{\ell}}^{1/2} \abs{\Group_{\ell'}}^{1/2}}
\end{align}
is non-zero only if there is at least one edge in $\set{E}$ joining the groups $\Group_{\ell}$ and $\Group_{\ell'}$. The Laplacian $\widetilde{\Lap}_{\ell\ell'}$ preserves the connections present in the original graph represented by $\Lap$. 

The dimension of the unknown vector in \eqref{eq:fast_decoder} is $\nbGroups$, which can be much smaller than $\nbVert$. This leads to a large gain in memory and computation time when either $\widetilde{\Lap}$ or matrix-vector multiplications with $\widetilde{\Lap}$ can be computed rapidly. If $g$ has a small degree, then $\widetilde{\Lap}$ can be computed explicitly in a short amount of time. In such a case, we can solve \eqref{eq:fast_decoder} faster than \eqref{eq:practical_decoder} as it involves matrices of (much) smaller dimensions. In general, it is however not always straightforward to find an a efficient implementation for matrix-vector multiplications with $\widetilde{\Lap}$ without temporary going back to the signal domain of dimension $\nbVert$, \ie, multiplying the vector $\tilde{\vec{z}}$ with $\Average^\adjoint$, filtering the high dimensional signal $\Average^\adjoint \tilde{\vec{z}}$, and downsampling the result. Even though solving \eqref{eq:fast_decoder} might still be faster than solving \eqref{eq:practical_decoder} in this situation, we loose part of the efficiency by working temporarily in the high dimensional domain. We thus have less flexibility in the choice of $g$ with this reconstruction technique.

Let us also mention that \eqref{eq:fast_decoder} can be used to initialise the algorithm used to solve \eqref{eq:practical_decoder} with a good approximate solution. As in multigrid approaches to solve linear systems of equations, see, \eg, \cite{borzi09, dreyer00}.

The following theorem bounds the error between the signal recovered by \eqref{eq:fast_decoder} and the original vector.

\begin{theorem}
\label{th:fast_decoder}
Let $\Omega = \{\omega_1, \ldots, \omega_\nbGroupsRed \}$ be a set of $\nbGroupsRed$ indices selected independently from $\{ 1, \ldots, \nbGroups\}$ using a sampling distribution $\prob \in \Rbb^\nbGroups$, $\Meas, \ma{P}, \widetilde{\Meas}, \widetilde{\ma{P}}$ be the associated matrices constructed respectively in \eqref{eq:subsampling_matrix_def}, \eqref{eq:subsampled_weight_matrix}, \eqref{eq:small_subsampling_matrix} and \eqref{eq:small_weight_matrix}, and $M_{\rm max}>0$ be a constant such that $\norm{\ma{P}\Meas}_2 \leq M_{\rm max}$. Let $\xi, \delta \in (0, 1)$ and suppose that $\nbGroupsRed$ satisfies \eqref{eq:sampling_condition}. With probability at least $1-\xi$, the following holds for all $\sig \in \PSignals$, all $\tilde{\err} \in \Rbb^{\nbGroupsRed}$, all $\reg > 0$, and all nonnegative nondecreasing polynomial functions $g$ such that $g(\eig_{\nbClass+1}) > 0$.

Let $\tilde{\sig}^*$ be the solution of \eqref{eq:fast_decoder} with $\tilde{\meas} = \widetilde{\Meas}\Average \sig + \tilde{\err}$. Define $\vec{\alpha}^* := \Fou_\nbClass \Fou_\nbClass^\adjoint \, {\Average}^\adjoint \tilde{\sig}^*$ and $\vec{\beta}^* := (\ma{I} - \Fou_\nbClass \Fou_\nbClass^\adjoint) \, {\Average}^\adjoint \tilde{\sig}^*$. Then,
\begin{align}
\label{eq:bound_alpha}
\norm{ \vec{\alpha}^* - \sig }_2
\; \leq \;
\inv{\sqrt{\nbGroupsRed (1-\delta)}} \; \cdot
\Bigg[ &
\left( 2 +  \frac{M_{\rm max}}{\sqrt{\reg g(\eig_{\nbClass+1})}}\right)\norm{\widetilde{\ma{P}} \tilde{\err}}_2
+ \left( M_{\rm max} \sqrt{\frac{g(\eig_\nbClass)}{g(\eig_{\nbClass+1})}} + \sqrt{\reg g(\eig_\nbClass)}\right) \norm{ \sig  }_2 \nonumber \\
& \left. + \; \epsilon \left( 2 M_{\rm max} + M_{\rm max} \sqrt{\frac{g(\eig_\nbVert)}{g(\eig_{\nbClass+1})}} + \sqrt{\reg g(\eig_\nbVert)} \right) \norm{ \sig  }_2 \right],
\end{align}
and
\begin{align}
\label{eq:bound_beta}
\norm{ \vec{\beta}^*}_2 
\leq \; \frac{1}{\sqrt{\reg g(\eig_{\nbClass+1})}} \norm{ \widetilde{\ma{P}} \tilde{\err} }_2 \; + \; \sqrt{\frac{g(\eig_\nbClass)}{g(\eig_{\nbClass+1})}} \norm{ \sig  }_2 \; + \; \epsilon \; \sqrt{\frac{g(\eig_\nbVert)}{g(\eig_{\nbClass+1})}} \norm{ \sig  }_2. 
\end{align}
\end{theorem}
\begin{proof}
See Appendix~\ref{app:proof_decoder}.
\end{proof}

The vector $\vec{\alpha}^*$ is the orthogonal projection of ${\Average}^\adjoint \tilde{\sig}^*$ onto $\spann(\Fou_\nbClass)$. The vector $\vec{\beta}^*$ is the projection of ${\Average}^\adjoint \tilde{\sig}^*$ onto the orthogonal complement of $\spann(\Fou_\nbClass)$. There are several remarks to make about the above theorem:
\begin{itemize}
\item Theorem~\ref{th:fast_decoder} shows that the result obtained via \eqref{eq:fast_decoder} is similar to the one we would have obtained by solving \eqref{eq:practical_decoder} -- see \cite{puy15b} for the error bounds -- with additional errors controlled by $\epsilon$. We recall that $\epsilon$ characterises how far $\sig$ is from a piecewise constant signal. As expected, the smaller $\epsilon$, the better the reconstruction.
\item The reconstruction quality improves when $g(\eig_\nbClass)$ and the ratio ${g(\eig_\nbClass)}/{g(\eig_{\nbClass+1})}$ go to $0$, and when ${g(\eig_\nbVert)}/{g(\eig_{\nbClass+1})}$ tends to $1$. We recall that we have $g(\eig_\nbVert) \geq g(\eig_{\nbClass+1}) > g(\eig_{\nbClass})$ by assumption.
\item The effect of the noise $\tilde{\err}$ decreases when $g(\eig_{\nbClass+1})$ increases, and, obviously, $\gamma$ should be adapted to the signal-to-noise ratio. 
\end{itemize}

Let us mention that the idea of ``coarsening'' a graph and the signals that live on it using a partition into different groups of nodes can also be found in~\cite{tremblay_arxiv2015}, where a multiresolution analysis method of graph signals is proposed. The coarsening method is however different than the one used here.

%
\section{Optimal sampling estimation}
\label{sec:estimate_sampling}

In this section, we come back to the sampling process of Section~\ref{sec:rip} and leave the reconstruction problem. We explain how to estimate the sampling distributions $\prob^*$ and $\vec{q}^*$ of Section~\ref{sec:rip} without computing the truncated Fourier matrix $\Fou_\nbClass$, as this computation is intractable for large graphs. The methods below only involve matrix-vector multiplications with the sparse Laplacian matrix $\Lap$ and are thus computationally tractable even for large $\nbVert$. 

\subsection{Estimation of $\prob^*$}

The distribution $\prob^*$, defined in (\ref{eq:pstar}), that minimises the coherence $\cumCoh_{\prob}$  is entirely defined by the values of $\|\SelectGroup^{(\ell)}\Fou_k\|_2$ for $\ell = 1, \ldots, \nbGroups$, which are thus the quantities we need to evaluate.

We recall that, in graph signal processing, a filter is represented by a function $h: \Rbb \rightarrow \Rbb$, and that the signal $\sig$ filtered by $h$ is
\begin{align}
\sig_{h} := \Fou \, \diag(\hat{\vec{h}}) \, \Fou^\adjoint \sig \in \Rbb^\nbVert,
\end{align}
where $\hat{\vec{h}} = (h(\eig_1), \ldots, h(\eig_\nbVert))^\adjoint \in \Rbb^{\nbVert}$. To filter the signal $\sig$ without actually computing the graph Fourier transform of $\sig$, we can approximate the function $h$ by a polynomial 
\begin{align}
r(t) := \sum_{i=0}^d \alpha_i \, t^i \approx h(t)
\end{align}
of degree $d$, and compute $\sig_{r}$ instead of $\sig_{h}$. The filtered signal $\sig_{r}$ is computed rapidly using the formula
\begin{align}
\sig_{r} = \sum_{i=0}^d \alpha_i \; \Fou \, \diag(\eig_1^i, \ldots, \eig_\nbVert^i) \, \Fou^\adjoint \sig = \sum_{i=0}^d \alpha_i \; \Lap^i \sig,
\end{align}
that involves only matrix-vector multiplications with the sparse Laplacian matrix $\Lap$. We let the reader refer to \cite{hammond11} for more information on this fast filtering technique. For any polynomial function $r$ of the form above and any matrix $\ma{A} \in \Rbb^{\nbVert \times \nbVert}$, we define
\begin{align}
r(\ma{A}) := \sum_{i=0}^d \alpha_i \; \ma{A}^i .
\end{align}
Note that $r(\Lap) = \Fou \, r(\Eig) \, \Fou^\adjoint$.

Let $i_{\eig_k} : \Rbb \rightarrow \Rbb$ be the ideal low-pass filter at cutoff frequency $\eig_k$, \ie, the filter that satisfies
\begin{align} 
i_{\eig_k}(t)
=
\left\{
\begin{array}{cc}
1 & \text{if } t \leq \eig_k,\\
0 & \text{otherwise}.
\end{array}
\right.
\end{align}
We have $\Fou_k \Fou_k^\adjoint = i_{\eig_k}\left(\Lap\right)$. Then, we notice that
\begin{align} 
\norm{\SelectGroup^{(\ell)}\Fou_k}_2^2
=
\norm{\SelectGroup^{(\ell)}\Fou_k \Fou_k^\adjoint {\SelectGroup^{(\ell)}}^\adjoint}_2
=
\norm{\SelectGroup^{(\ell)} \; i_{\eig_k}\left(\Lap\right) \; {\SelectGroup^{(\ell)}}^\adjoint}_2.
\end{align}
We recall that $\SelectGroup^{(\ell)}$ is the matrix that restricts the signal to the nodes belonging to $\Group_{\ell}$. The matrix appearing on the right hand side of the last equality corresponds to the linear operator that 1) extends a vector on the complete graph by inserting $0$ in all groups $\ell' \neq \ell$, 2) low-pass filters the extended signal, 3) restricts the result to the group $\Group_{\ell}$. This process can be approximated by replacing the ideal low-pass filter $i_{\eig_k}$ with a polynomial approximation $\tilde{i}_{\eig_k}$ of $i_{\eig_k}$ and 
\begin{align}
\norm{\SelectGroup^{(\ell)} \; i_{\eig_k}\left(\Lap\right) \; {\SelectGroup^{(\ell)}}^\adjoint}_2
\approx
\norm{\SelectGroup^{(\ell)} \; \tilde{i}_{\eig_k}\left(\Lap\right) \; {\SelectGroup^{(\ell)}}^\adjoint}_2.
\end{align}
To estimate $\prob^*$, we estimate the spectral norm of the matrix appearing on the right hand side for which matrix-vector multiplication is fast. The quality of the approximation depends obviously on the choice of the polynomial $\tilde{i}_{\eig_k}$.

Estimating $\| \SelectGroup^{(\ell)} \; \tilde{i}_{\eig_k}\left(\Lap\right) \; {\SelectGroup^{(\ell)}}^\adjoint \|_2$ amounts to computing the largest eigenvalue of $\SelectGroup^{(\ell)} \; \tilde{i}_{\eig_k}\left(\Lap\right) \; {\SelectGroup^{(\ell)}}^\adjoint$ which can be done, \eg, by using the power method. This method requires matrix-vector multiplication only with $\SelectGroup^{(\ell)}$ and $\tilde{i}_{\eig_k}\left(\Lap\right)$ and is thus fast.
Finally, the approximation $\bar{\vec{p}} \in \Rbb^{\nbGroups}$ of $\prob^*$ satisfies
\begin{align} 
\bar{\prob}_{\ell} := \frac{\lambda_{\rm max}(\SelectGroup^{(\ell)} \; \tilde{i}_{\eig_k}\left(\Lap\right) \; {\SelectGroup^{(\ell)}}^\adjoint)}{\sum_{\ell'=1}^\nbGroups \lambda_{\rm max}(\SelectGroup^{(\ell')} \; \tilde{i}_{\eig_k}\left(\Lap\right) \; {\SelectGroup^{(\ell')}}^\adjoint)}.
\end{align}

Note that an estimation of $\eig_\nbClass$ is required beforehand to define the filter $\tilde{i}_{\eig_k}$. We estimate this value using the dichotomy method presented in~\cite{puy15b}.

\subsection{Estimation of $\vec{q}^*$}

Computing $\bar{\prob}$ requires the estimation of $\nbGroups$ eigenvalues. Even though these estimations can be done in parallel, this process might still be too slow for certain applications. As explained before, when $\nbClass$ is small, we can use the sampling distribution $\vec{q}^*$ in (\ref{eq:qstar}) that minimises $\bar{\cumCoh}_{\vec{p}}$. This distribution is faster to compute than $\bar{\prob}$.

We start by noticing that we have
\begin{align}
\norm{\SelectGroup^{(\ell)}\Fou_k}_F^2 = \sum_{i \in \Group_{\ell}} \norm{\Fou_\nbClass^\adjoint \vec{\delta}_i}_2^2
\end{align}
for each group $\ell = 1, \ldots, \nbGroups$. The vector $\vec{\delta}_i \in \Rbb^\nbVert$ is the unit vector that is null on all nodes except at node $i$. Hence, we only need an estimation of $\|\Fou_\nbClass^\adjoint \vec{\delta}_i\|_2^2$, $i = 1, \ldots, \nbVert$, to estimate $\vec{q}^*$. An algorithm was already proposed in~\cite{puy15b} to estimate these values. We let the reader refer to Algorithm $1$ in~\cite{puy15b} for the details of the method. We just recall that this estimation is obtained by filtering $O(\log(\nbVert))$ random signals with a polynomial approximation of $i_{\eig_\nbClass}$. Finally, our estimation $\bar{\vec{q}} \in \Rbb^{\nbGroups}$ of $\vec{q}^*$ has entries
\begin{align} 
\bar{\vec{q}}_{\ell} := \frac{\sum_{i \in \Group_{\ell}} \norm{\Fou_\nbClass^\adjoint \vec{\delta}_i}_F^2}{\sum_{\ell=1}^\nbGroups \sum_{i \in \Group_{\ell}} \norm{\Fou_\nbClass^\adjoint \vec{\delta}_i}_F^2},
\end{align}
where each $\|\Fou_\nbClass^\adjoint \vec{\delta}_i\|_2^2$ are estimated by Algorithm $1$ in~\cite{puy15b}.

This estimation is faster than for $\bar{\vec{\prob}}$ because the power method is an iterative method that involves one filtering at each iteration. Furthermore, the power method is run independently for each group $\Group_{\ell}$. In total, (much) more than $\Group_{\ell}$ filterings are thus required. On the contrary, for $\bar{\vec{q}}$, we just need to filter $O(\log(\nbVert))$ signals to obtain the estimation. In most situations, we already have $O(\log(\nbVert)) \leq \Group_{\ell}$ and computing $\bar{\vec{q}}$ is thus faster than computing $\bar{\prob}$.

%
\section{Experiments}
\label{sec:experiments}

In this last section, we first test our sampling strategies on two different graphs to illustrate the effect of the different sampling distributions on the minimum number of samples required to ensure that the RIP holds. Then, we apply our sampling strategy for user-guided object segmentation. In this application, we also test the different recovery techniques proposed in Section~\ref{sec:fast_decoder}.

\subsection{Sampling distributions}

\begin{figure*}
\centering
\includegraphics[width=.24\linewidth]{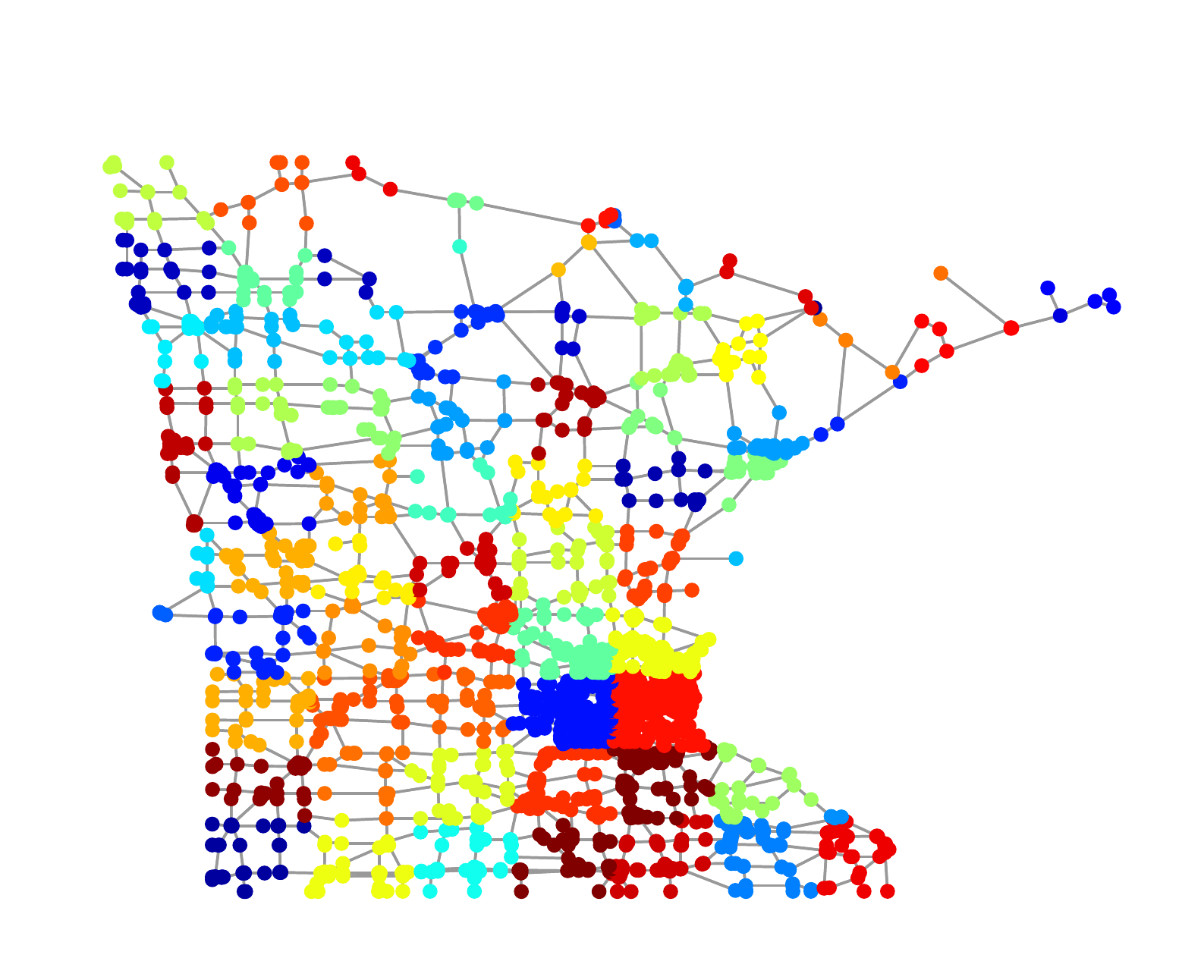}
\includegraphics[width=.24\linewidth]{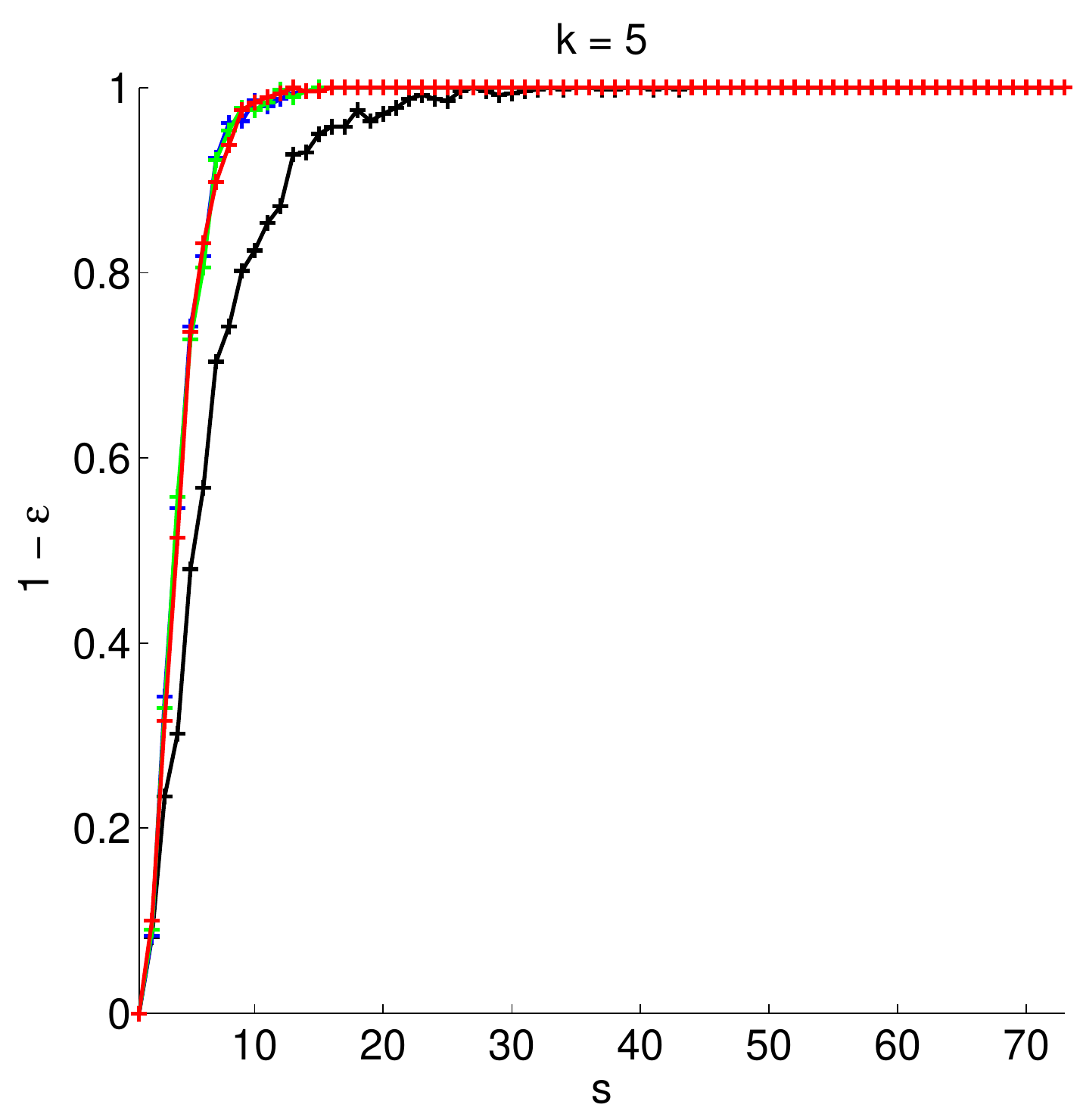}
\includegraphics[width=.24\linewidth]{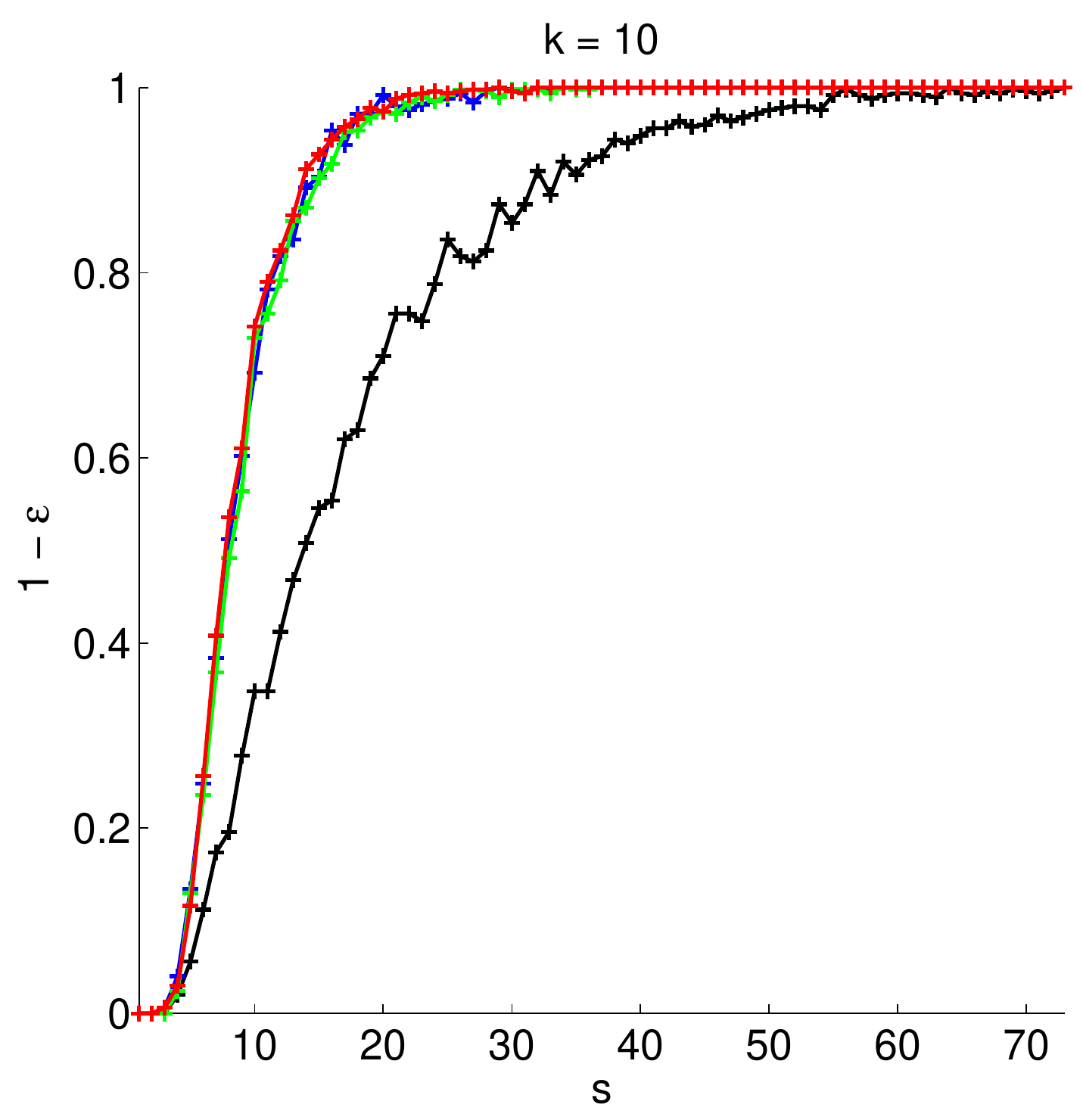}
\includegraphics[width=.24\linewidth]{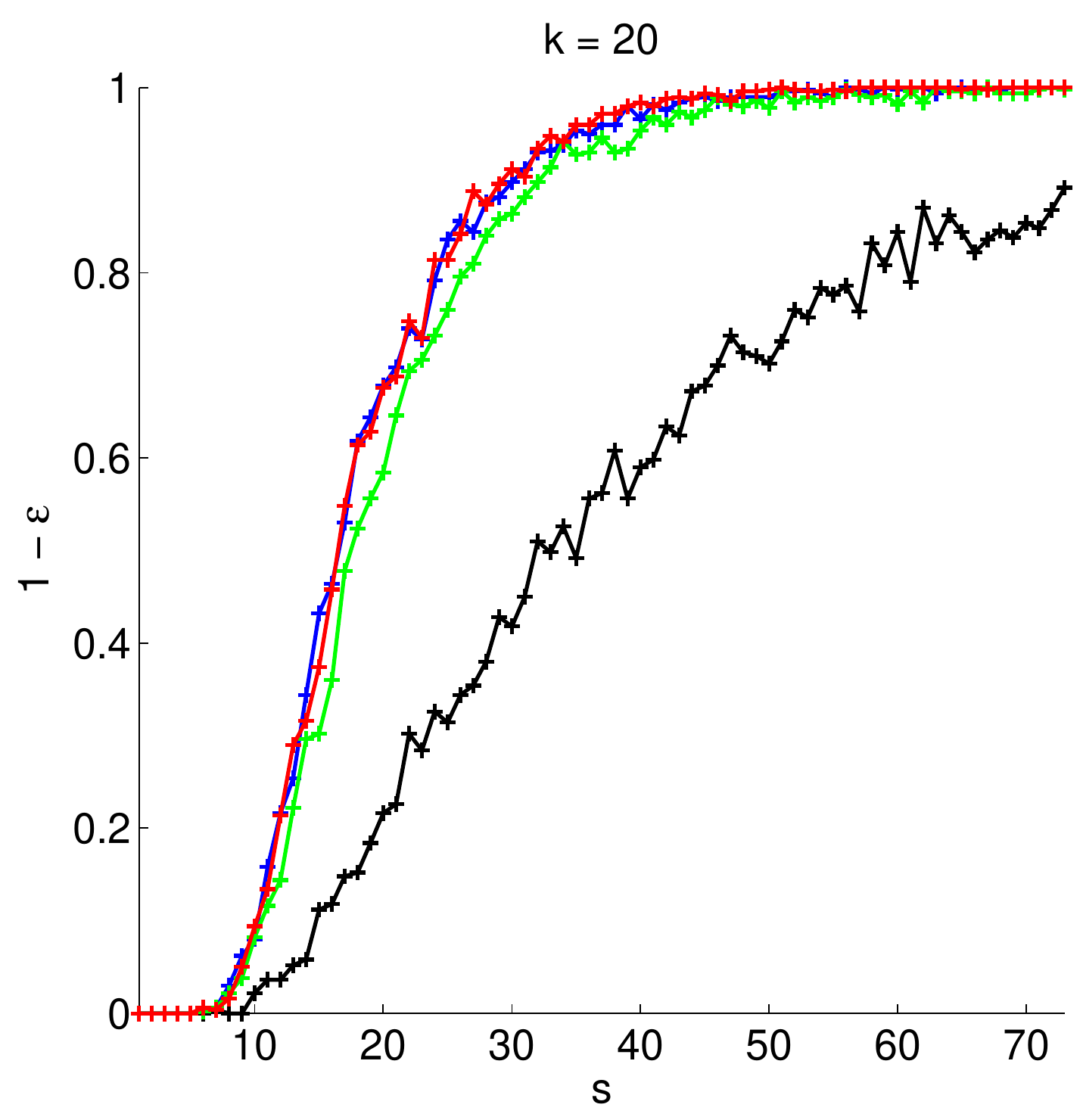}
\\
\includegraphics[width=.24\linewidth]{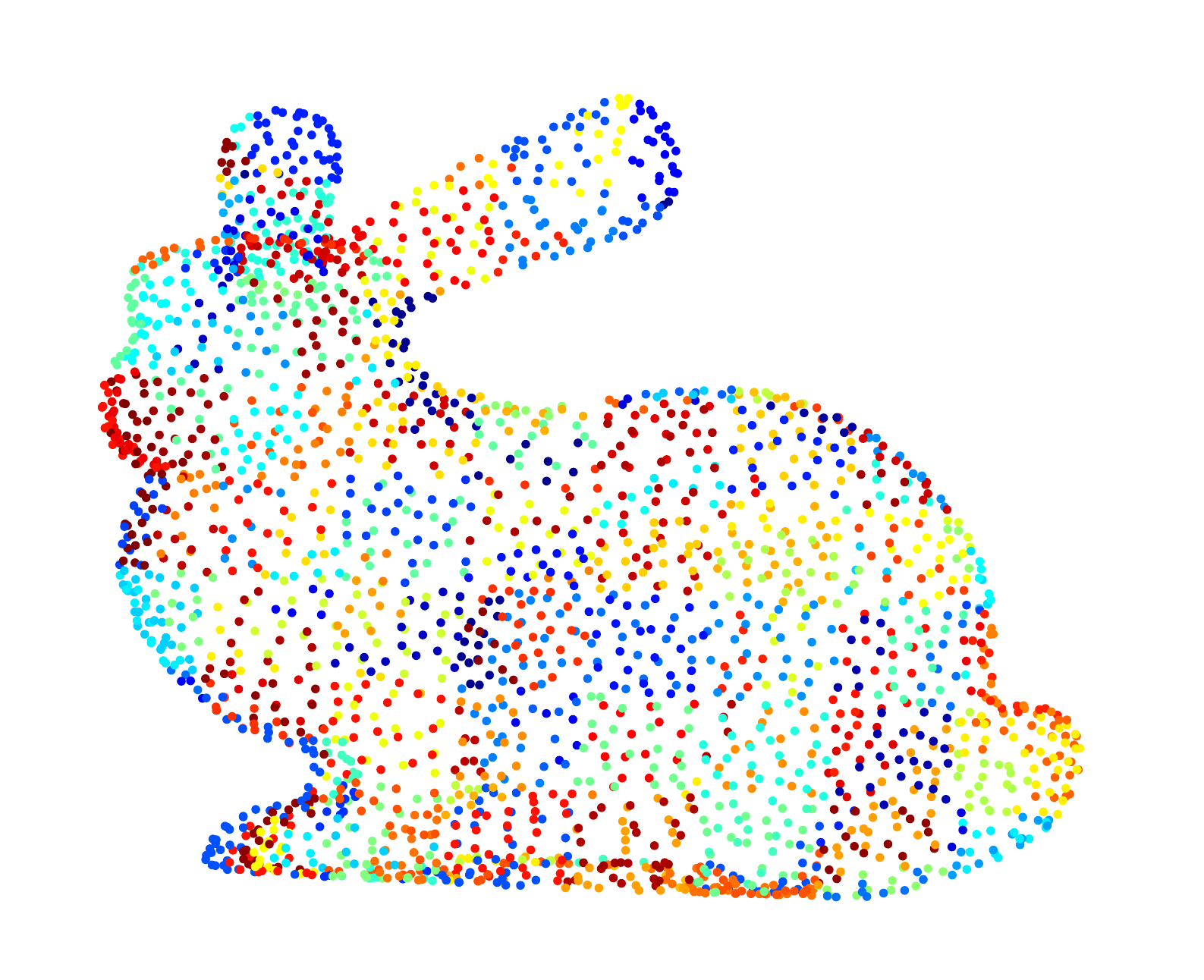}
\includegraphics[width=.24\linewidth]{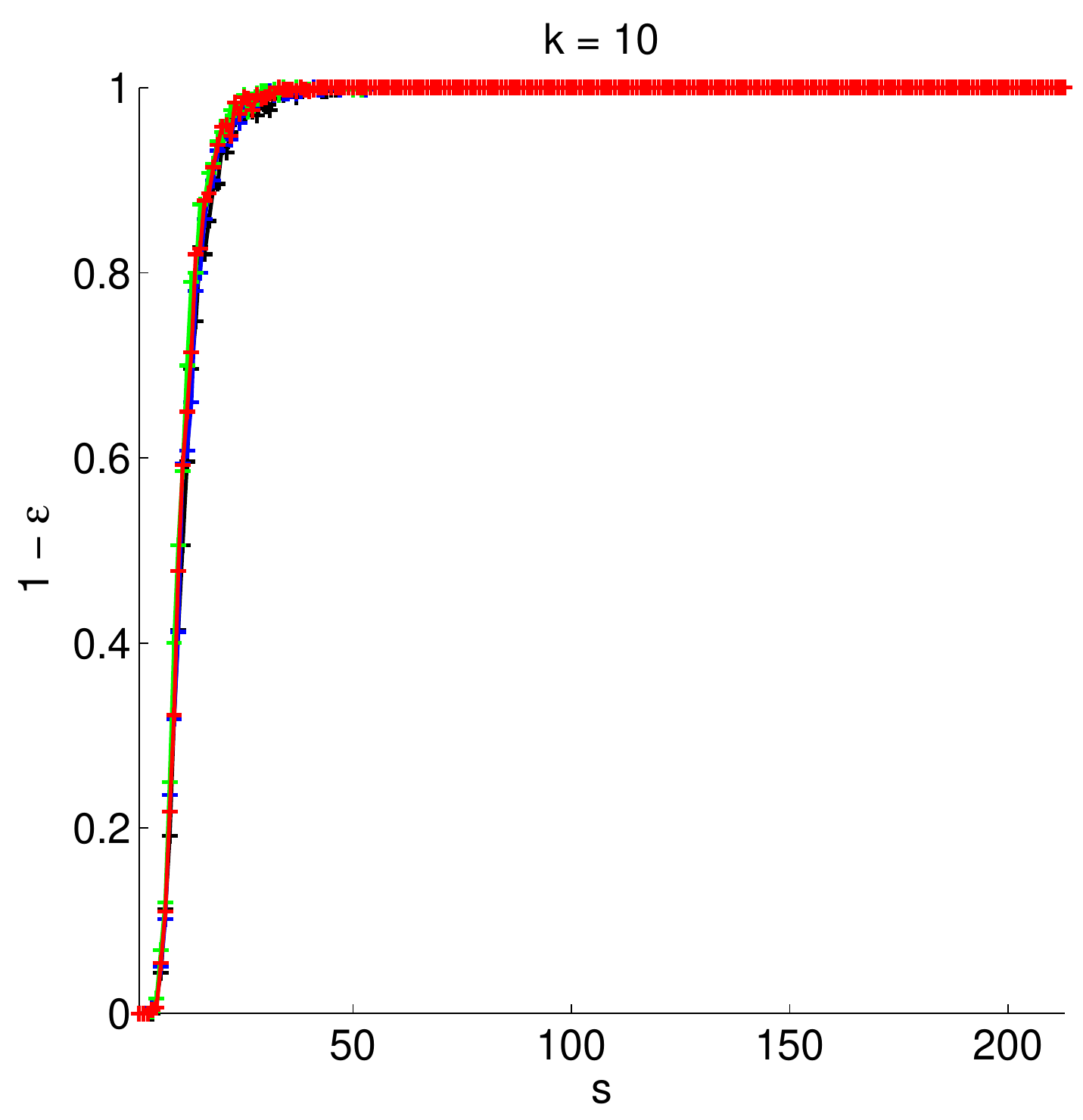}
\includegraphics[width=.24\linewidth]{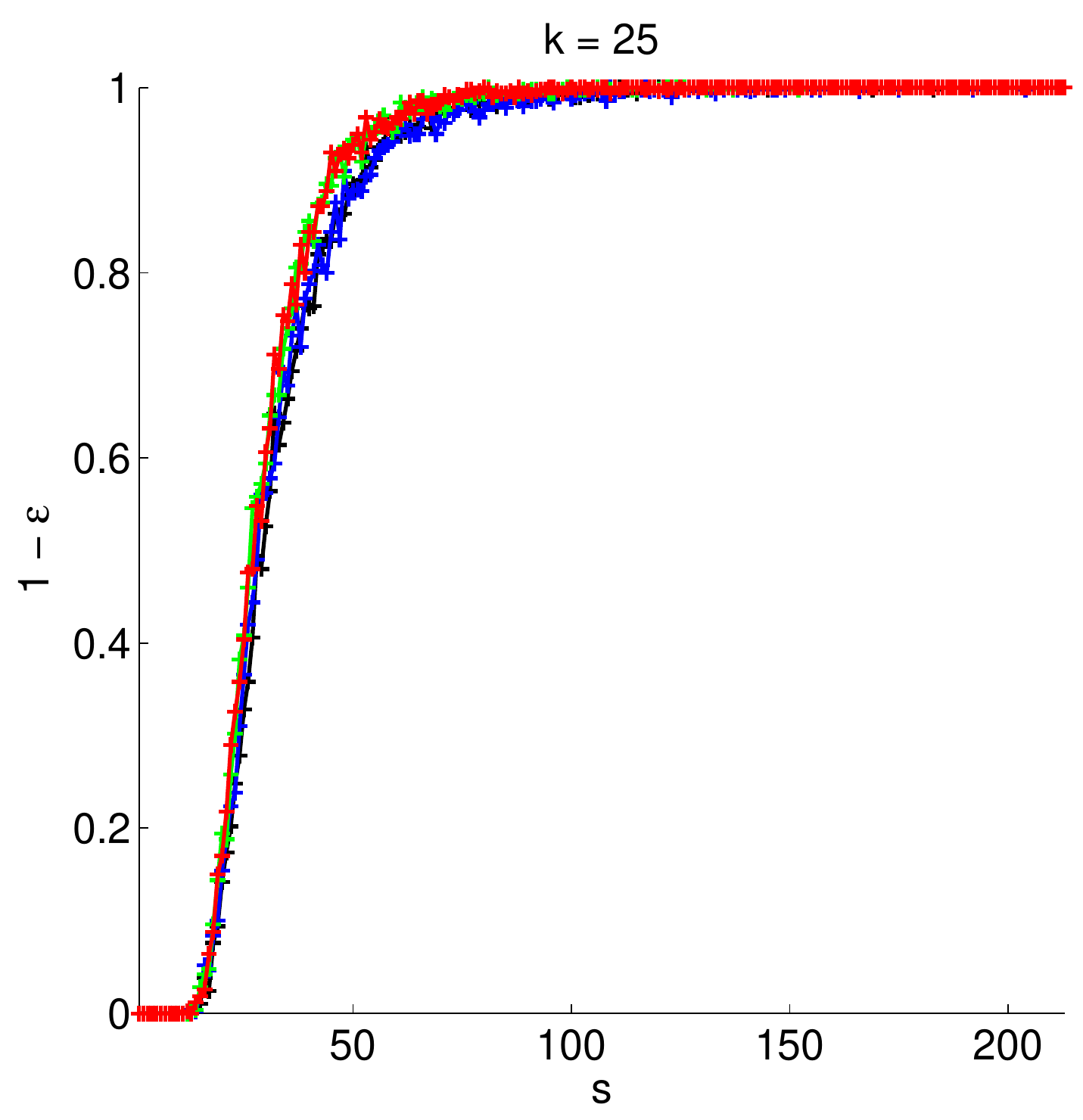}
\includegraphics[width=.24\linewidth]{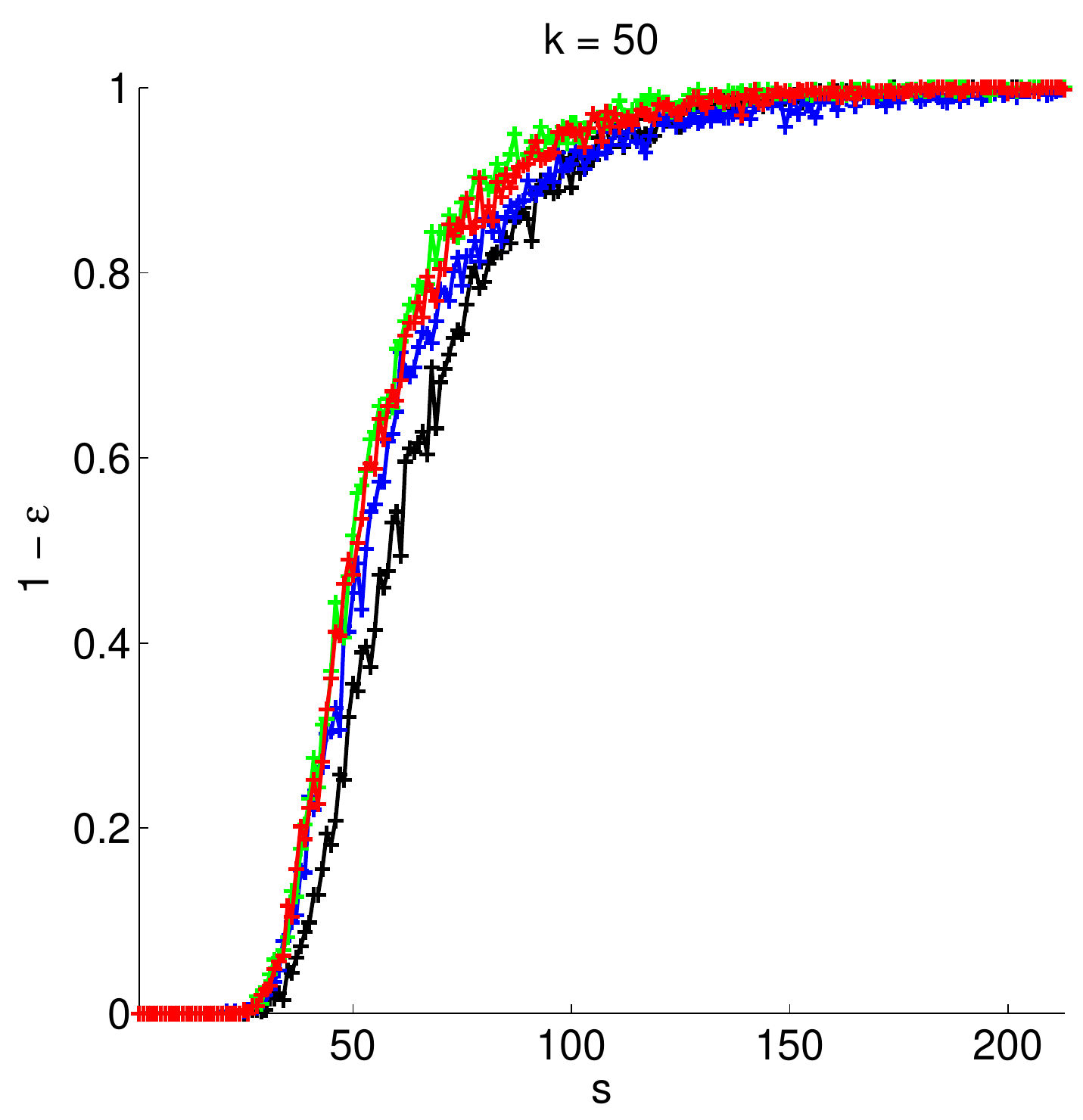}
\caption{\label{fig:graphs} First column from the left: Minnesota graph (top); bunny graph (bottom). The groups $\Group_1, \ldots, \Group_\nbGroups$ are indicated by different colors. Other columns: probability that $\underline{\delta}_{k}$ is less that $0.995$ as a function of $\nbGroupsRed$. The curves in black, red, blue and green are obtained using the sampling distributions $\vec{u}$, $\vec{p}^*$, $\bar{\vec{q}}$ and $\bar{\prob}$, respectively. The top row shows the result for the Minnesota graph. The bottom row shows the result for the bunny graph. The bandlimit $k$ is indicated on top of each curve.}
\end{figure*}

We perform experiments on two different graphs: the Minnesota graph of size $\nbVert = 2642$ and the bunny graph of size $\nbVert = 2503$. Both graphs are presented in Fig.~\ref{fig:graphs} and are available in the GSP toolbox \cite{perraudin14}. For each graph, we group the nodes using the spatial coordinates associated to each node. For the Minnesota graph, we divide the space into $100$ cells and group the nodes that fall in the same cell. After removing empty cells, we obtain the $\nbGroups = 73$ groups represented in Fig.~\ref{fig:graphs}. For the bunny graph, we obtain $\nbGroups = 213$ groups with a similar procedure (see Fig.~\ref{fig:graphs}).

For each graph, we compute the combinatorial Laplacian and $\Fou_{k}$ for different values of $k$. Then, we compute the lower RIP constant, \ie, the constant $\underline{\delta}_k>0$ that satisfies
\begin{align*}
\underline{\delta}_{k} 
= 1 - \inv{\nbGroupsRed} \;\;
\inf_{\substack{\sig \in \spann(\Fou_{\nbClass}) \\ \norm{x}_2 = 1}}  
\;\norm{\ma{P} \Meas \; \sig}_2^2.
\end{align*}
This constant is the smallest value that $\delta$ can take such that the left-hand side of the RIP \eqref{eq:RIP} holds. Remark that
\begin{align}
\underline{\delta}_{k} = 1 - \inv{\nbGroupsRed} \; \lambda_{\rm min} \left( \Fou_{k}^\adjoint \Meas^\adjoint \ma{P}^2 \Meas \Fou_{k}\right).
\end{align}
We estimate $\underline{\delta}_{k}$ for $500$ independent draws of the set $\Omega$, which defines the matrices $\ma{P}\Meas$, and different numbers of selected groups $\nbGroupsRed$. All samplings are done in the conditions of Theorem~\ref{th:rip} using the sampling distributions $\vec{u}, \vec{p}^*, \bar{\prob}$ and $\bar{\vec{q}}$. The vector $\vec{u}$ denotes the uniform distribution over $\{1, \ldots, \nbGroups\}$. When conducting this experiment with the estimated distributions $\bar{\prob}$ and $\bar{\vec{q}}$, we re-estimate these distributions at each of the $500$ trials. These distributions are estimated using Jackson-Chebychev polynomials of order $50$ \cite{napoli13}. For the Minnesota graph, we consider the bandlimits $\nbClass = 5, 10, 20$. For the bunny graph, we consider the bandlimits $\nbClass = 10, 25, 50$.

We present the probability that $\underline{\delta}_{k}$ is less than $0.995$, estimated over the $500$ draws of $\Omega$, as a function of $\nbGroupsRed$ in Fig.~\ref{fig:graphs}. For the Minnesota graph, the performance is better when using the optimal distribution $\prob^*$ than when using the uniform distribution $\vec{u}$ for all $k$, which is in line with the theory. The estimated $\bar{\prob}$ and $\bar{\vec{q}}$ yield performance equivalent to $\prob^*$. This confirms that we can achieve similar sampling performance without having to compute the Fourier matrix $\Fou_\nbClass$, which, we recall, is intractable for large graphs. This also shows that $\bar{\vec{q}}$ can lead to nearly optimal results. For the bunny graph, all sampling distributions yield essentially the same results at all bandlimits. We notice a slight improvement at $k=50$ when using $\bar{\prob}$, $\bar{\vec{q}}$ or $\prob^*$ instead of $\vec{u}$.

For illustration, we present in Fig.~\ref{fig:sampling_distributions} examples of computed sampling distributions $\vec{p}^*$, $\bar{\prob}$ and $\bar{\vec{q}}$. All sampling distributions exhibit similar structures, which explains why they all yield about the same performance in our experiments.

\begin{figure*}
\centering
\begin{minipage}{.32\linewidth} \centering \scriptsize $\prob^*$ \end{minipage}
\begin{minipage}{.32\linewidth} \centering \scriptsize $\bar{\prob}$ \end{minipage}
\begin{minipage}{.32\linewidth} \centering \scriptsize $\bar{\vec{q}}$ \end{minipage}
\\
\includegraphics[height=40.5mm]{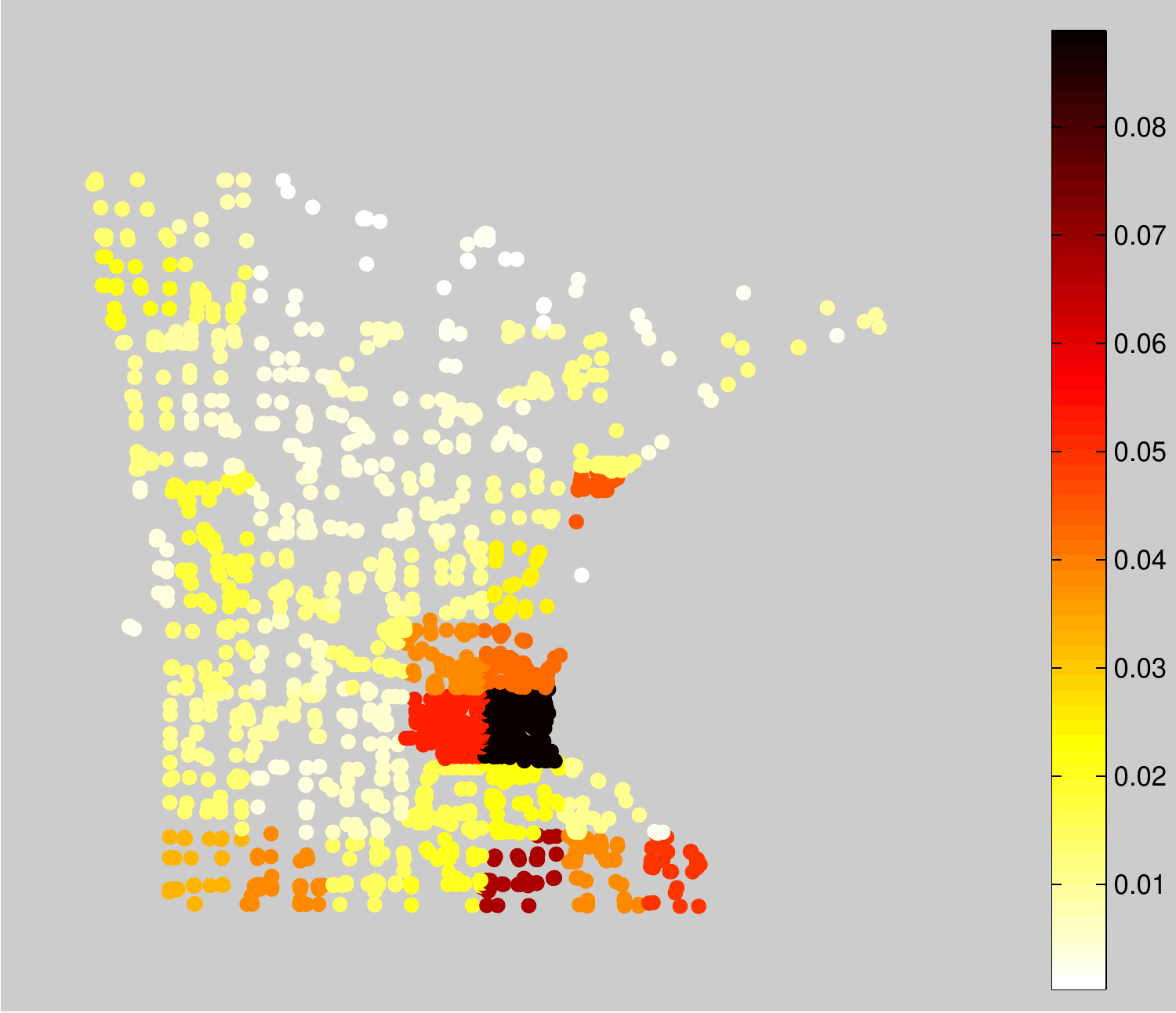}
\includegraphics[height=40.5mm]{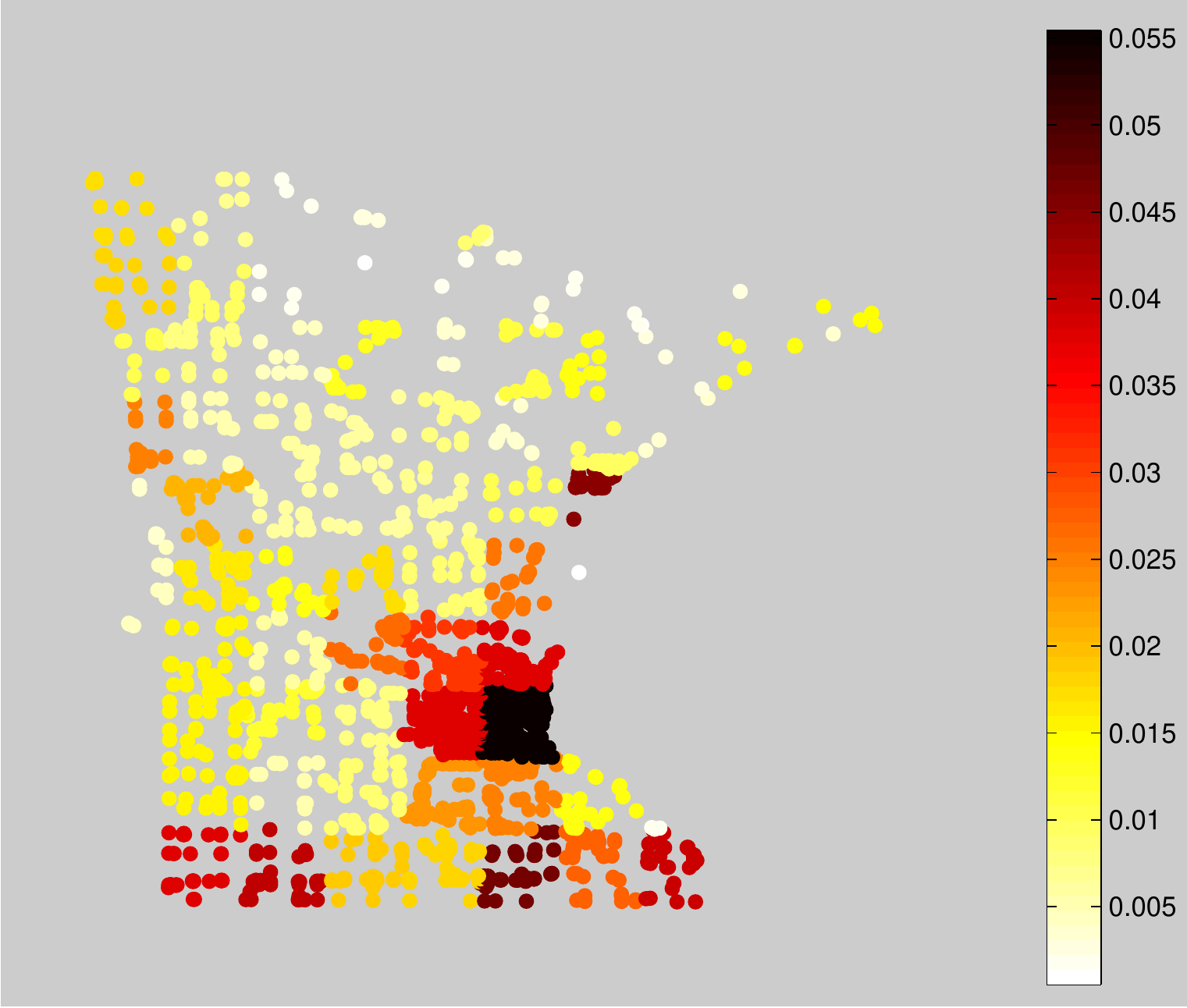}
\includegraphics[height=40.5mm]{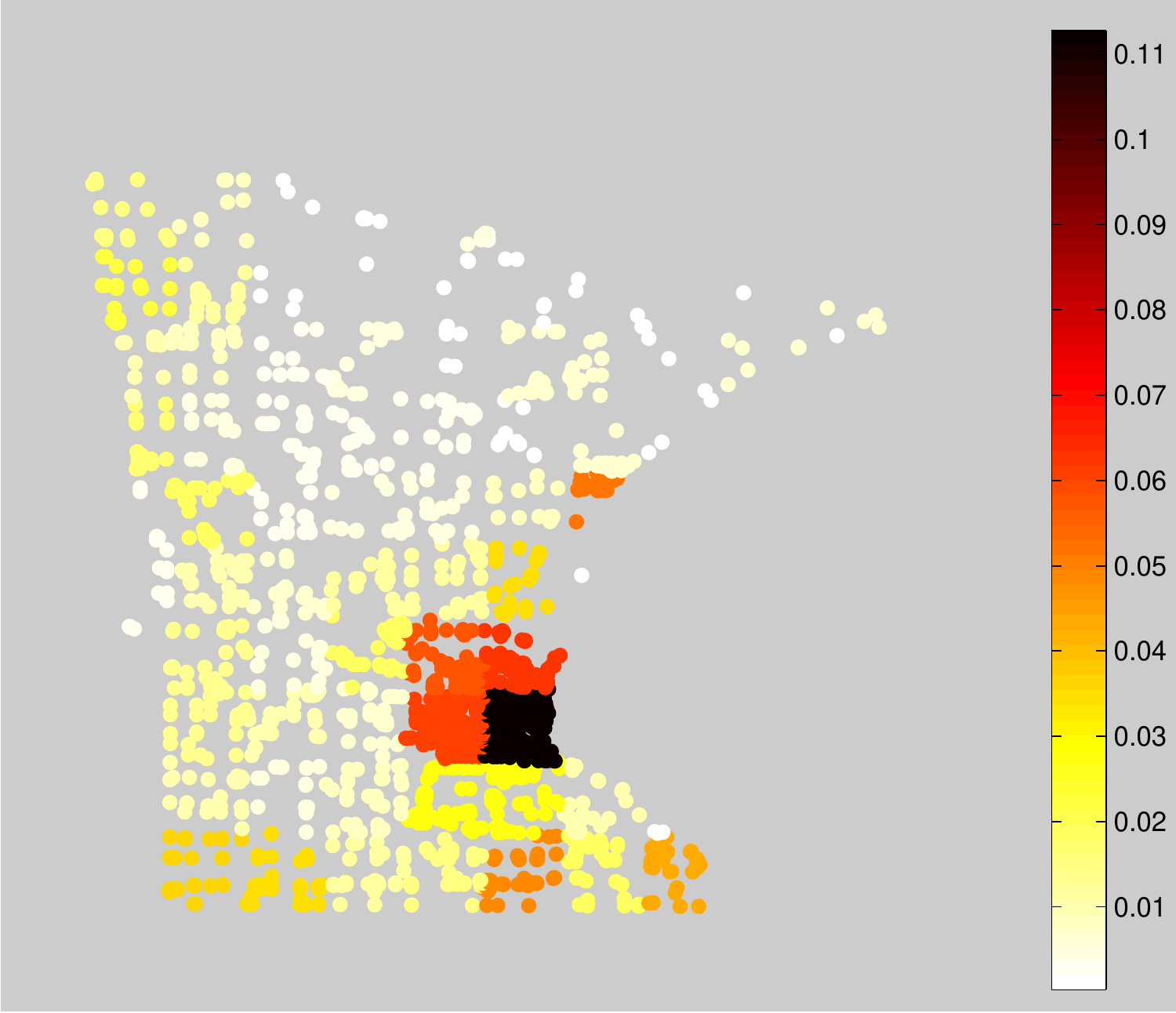}
\\
\includegraphics[height=40mm]{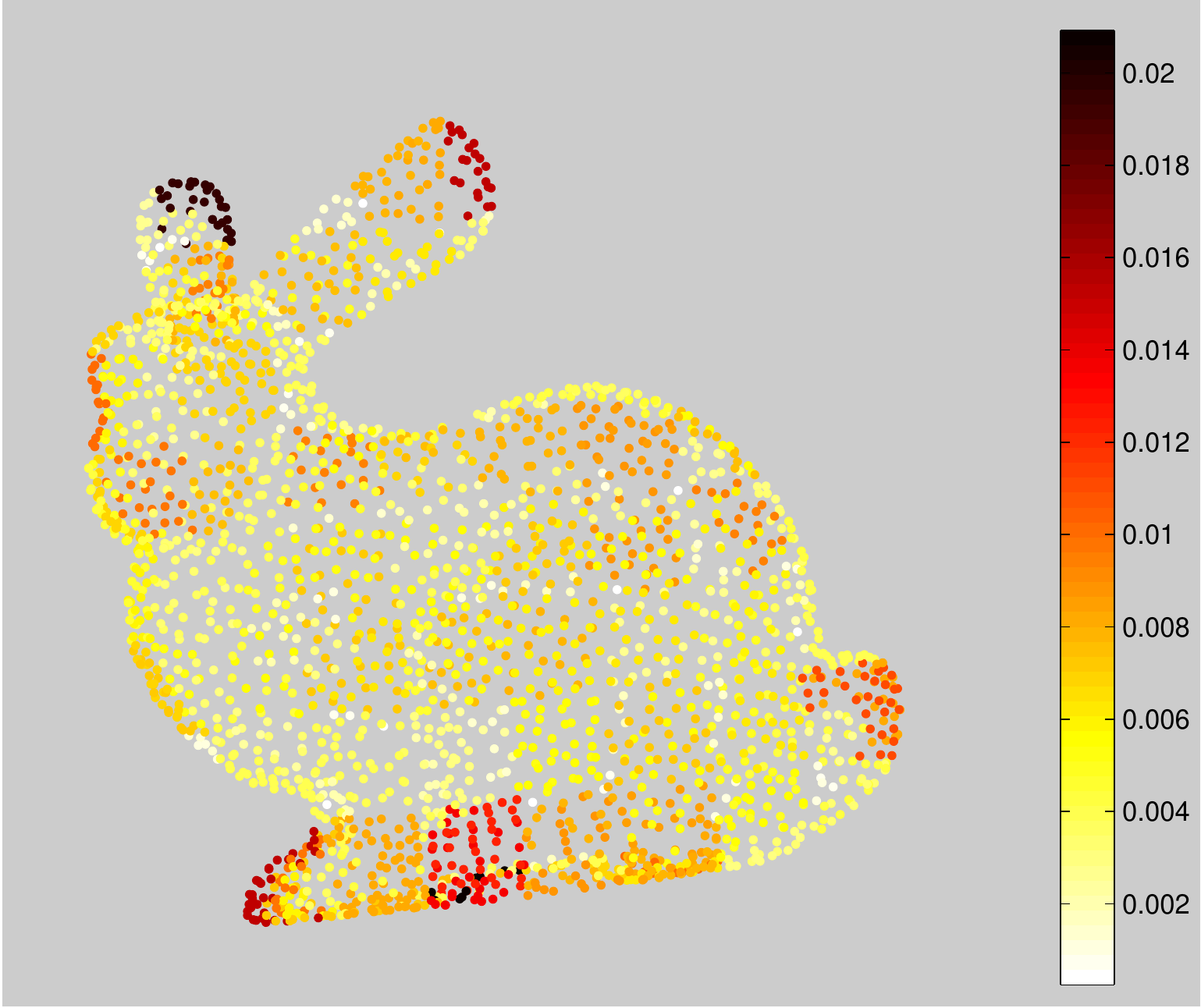}
\includegraphics[height=40mm]{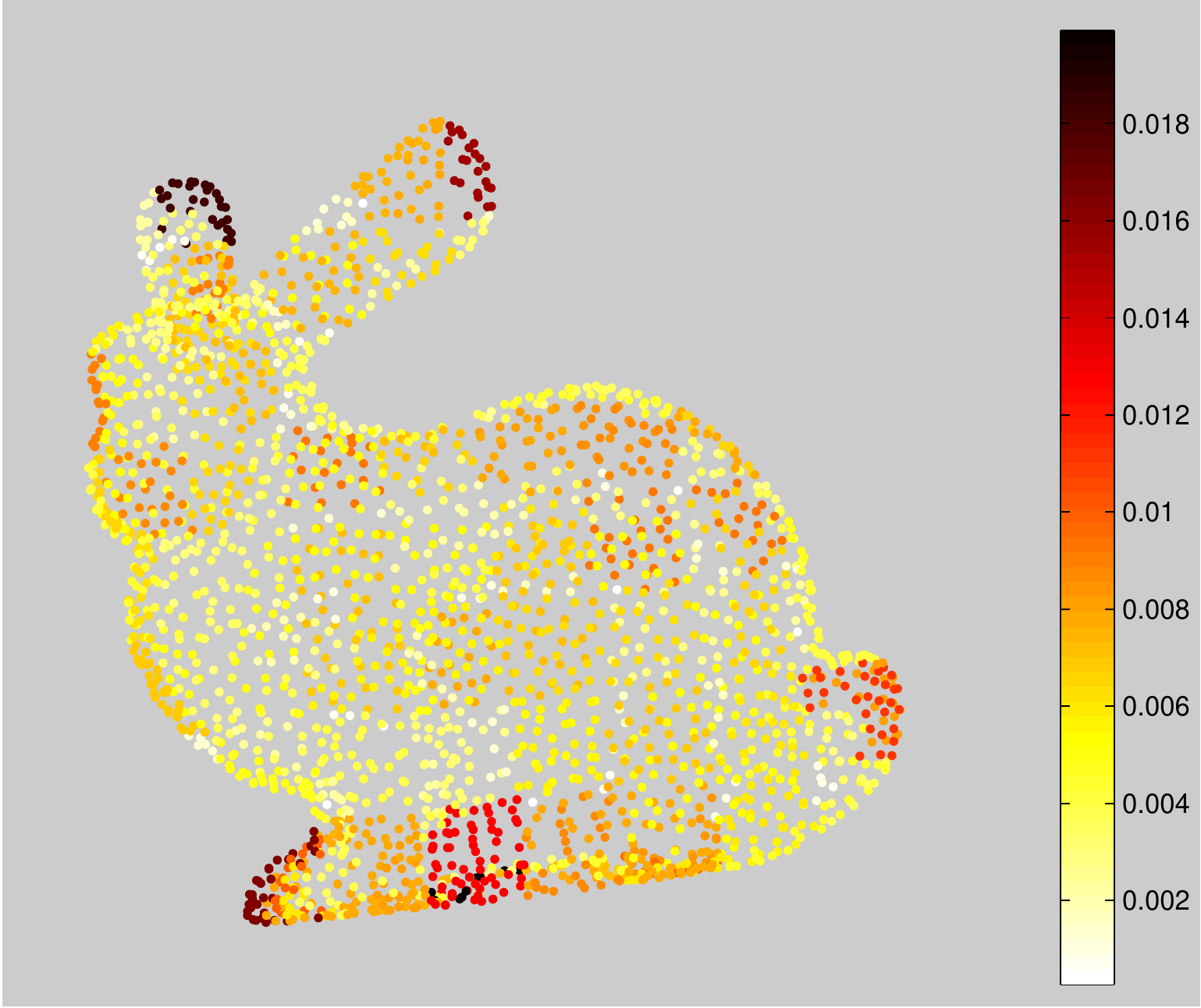}
\includegraphics[height=40mm]{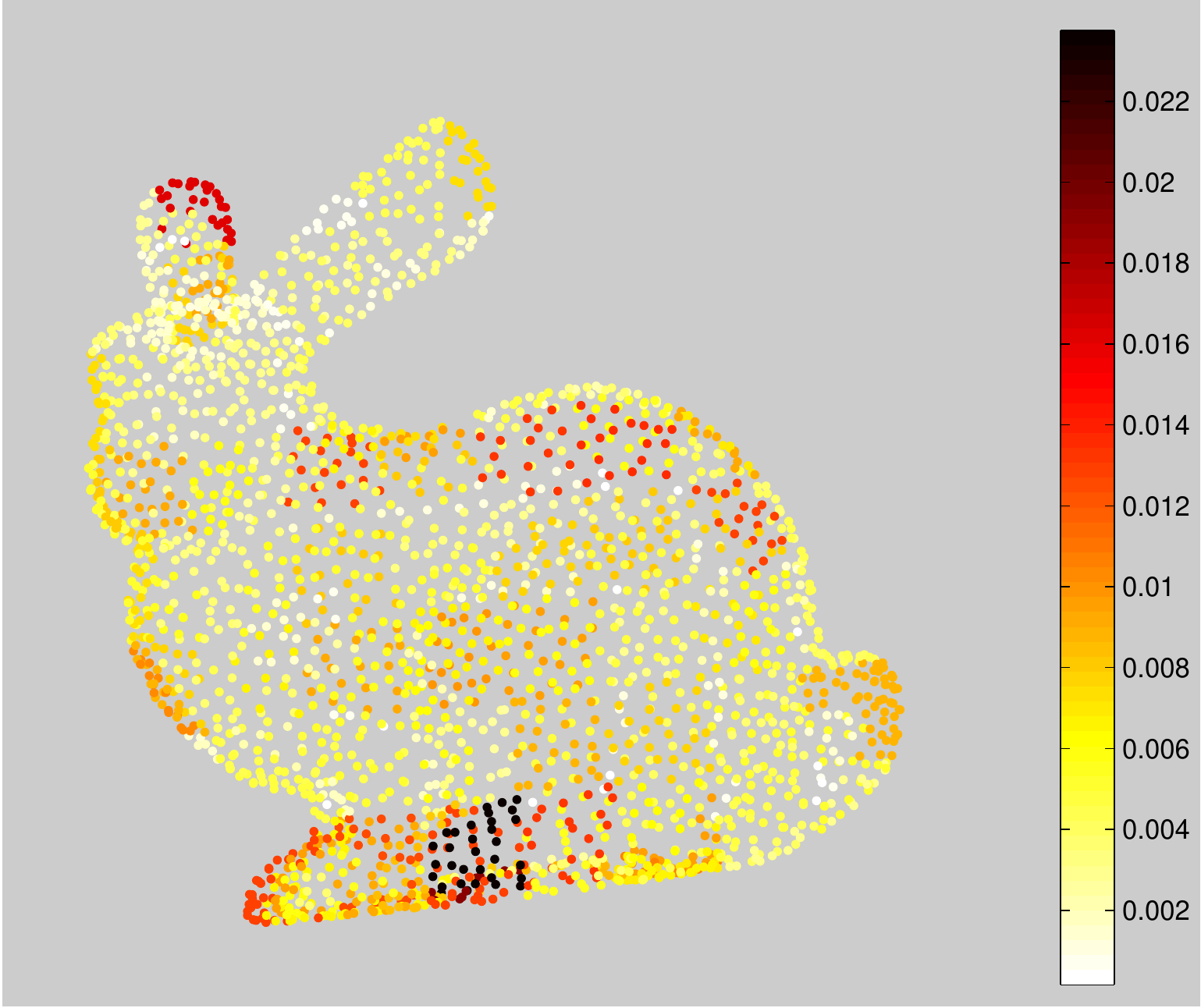}
\caption{\label{fig:sampling_distributions}Example of sampling distributions. Top panels: $\prob^*$ (left), $\bar{\prob}$ (middle), and $\bar{\vec{q}}$ (right) for the Minnesota graph at $\nbClass=10$. Bottom panels: $\prob^*$ (left), $\bar{\prob}$ (middle), and $\bar{\vec{q}}$ (right) for the bunny graph at $\nbClass=25$.}
\end{figure*}
%

\subsection{Object segmentation}

\subsubsection{Protocol}
\begin{figure*}
\centering
\begin{minipage}{.32\linewidth} \centering \scriptsize Original image \end{minipage}
\begin{minipage}{.32\linewidth} \centering \scriptsize $\norm{\SelectGroup^{(\ell)}\Fou_{\nbClass_0}}_2^2$ values \end{minipage}
\begin{minipage}{.32\linewidth} \centering \scriptsize $\norm{\SelectGroup^{(\ell)}\Fou_{\nbClass_0}}_F^2$ values \end{minipage}
\\
\begin{sideways}\color{white} \scriptsize Top \end{sideways}
\includegraphics[width=.32\linewidth]{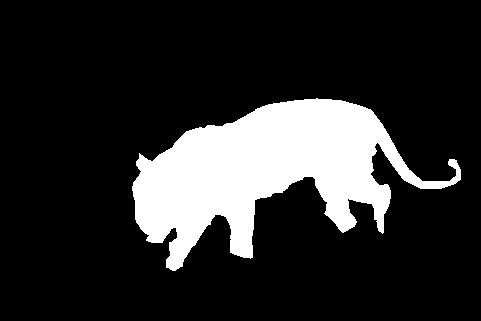}
\includegraphics[width=.32\linewidth]{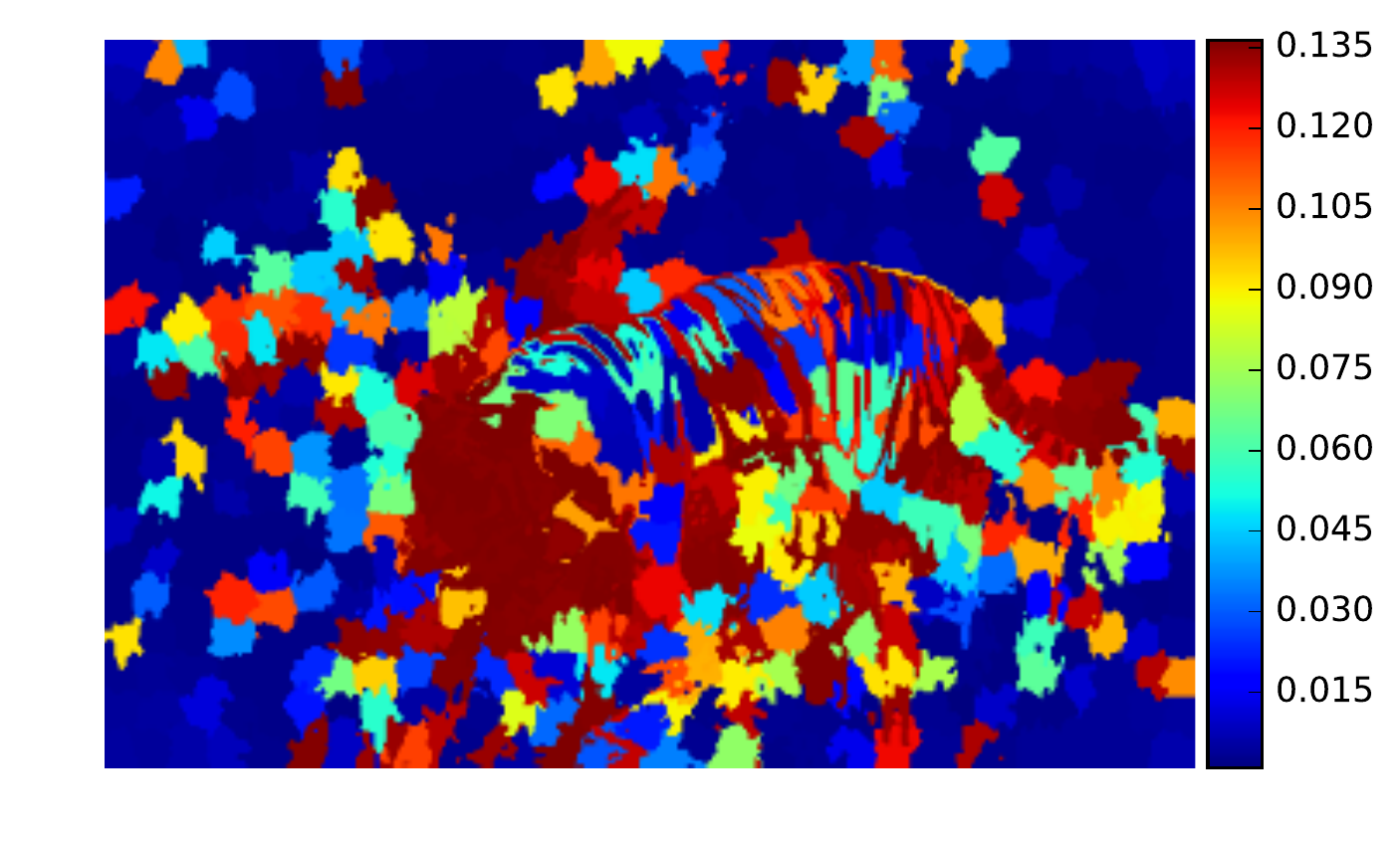}
\includegraphics[width=.32\linewidth]{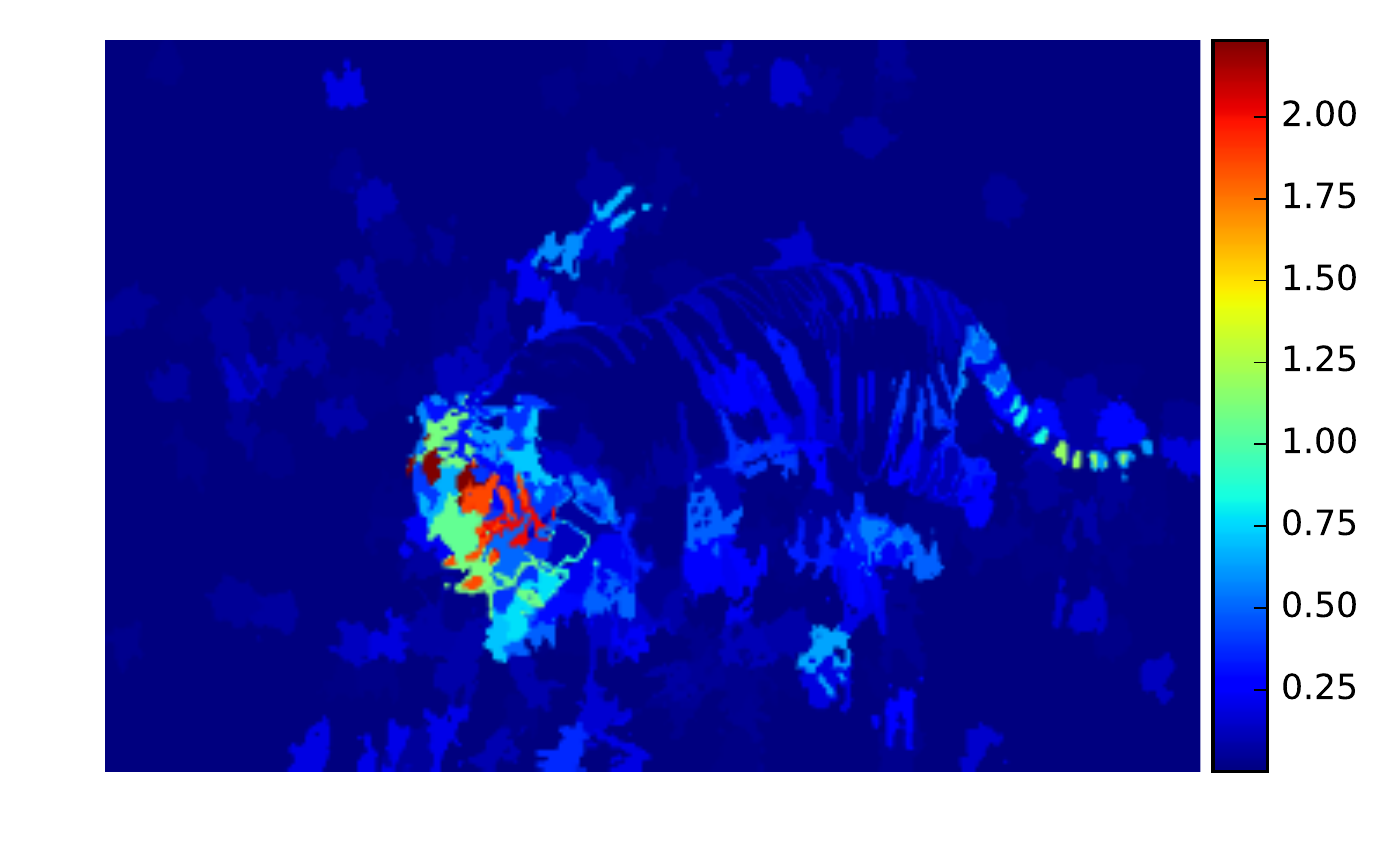}
\\
\begin{minipage}{.32\linewidth} \centering \scriptsize $\vec{u}$ \end{minipage}
\begin{minipage}{.32\linewidth} \centering \scriptsize $\bar{\vec{\prob}}$ \end{minipage}
\begin{minipage}{.32\linewidth} \centering \scriptsize $\bar{\vec{q}}$ \end{minipage}
\\
\begin{sideways} \hspace{4mm} \scriptsize Result solving~\eqref{eq:fast_decoder} \end{sideways}
\includegraphics[width=.32\linewidth]{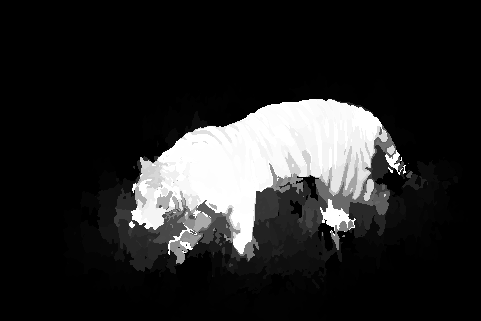}
\includegraphics[width=.32\linewidth]{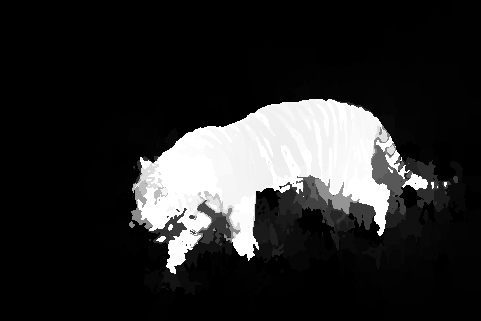}
\includegraphics[width=.32\linewidth]{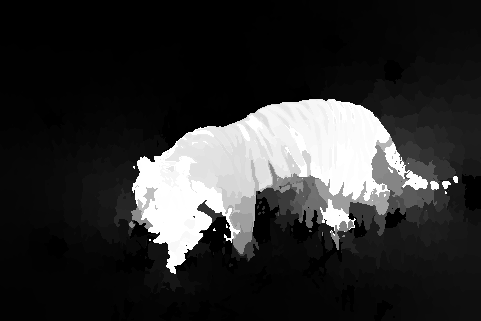}
\\
\begin{sideways} \hspace{4mm} \scriptsize Result solving~\eqref{eq:practical_decoder} \end{sideways}
\includegraphics[width=.32\linewidth]{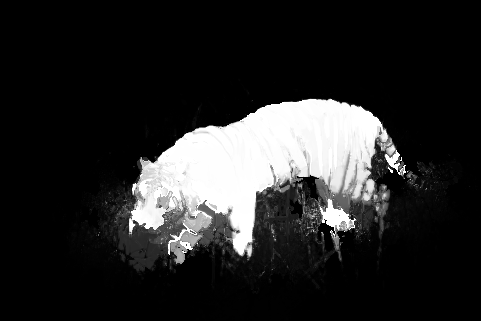}
\includegraphics[width=.32\linewidth]{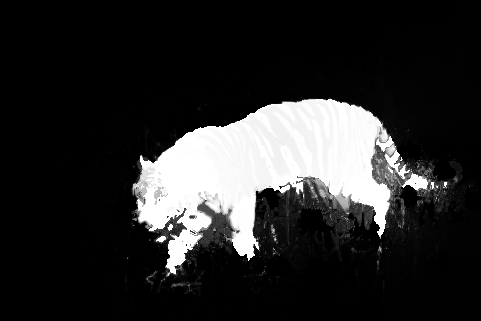}
\includegraphics[width=.32\linewidth]{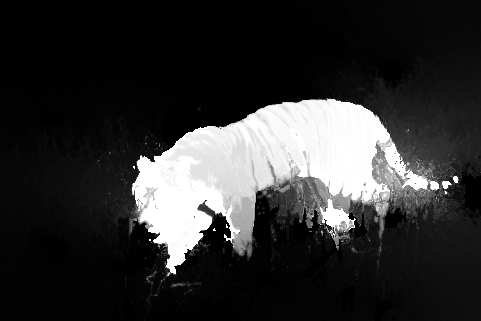}
\caption{\label{fig:segmentation_result}Top row: ground truth segmentation map (left); local group coherence map $\|\SelectGroup^{(\ell)}\Fou_{\nbClass_0}\|_2^2$ evaluated with the spectral norm (middle); local group coherence map $\|\SelectGroup^{(\ell)}\Fou_{\nbClass_0}\|_F^2$ evaluated with the Frobenius norm (right). Middle row: segmentation map estimated from $s=50$ sampled superpixels obtained by solving~\eqref{eq:fast_decoder} and using a uniform sampling (left), $\bar{\prob}$ (middle), $\bar{\vec{q}}$ (right). Bottom row: segmentation map estimated from $s=50$ sampled superpixels obtained by solving~\eqref{eq:practical_decoder} and using a uniform sampling (left), $\bar{\prob}$ (middle), $\bar{\vec{q}}$ (right).}
\end{figure*}

We now test our method for interactive object segmentation. We consider the image of size $\nbVert = 321 \times 481$ presented in Fig.~\ref{fig:image_example} for which our goal is to segment the tiger. The ground truth segmentation map $\sig \in \{0, 1\}^\nbVert$ is presented in Fig.~\ref{fig:segmentation_result}. The value $1$ (white) indicates the presence of the tiger and the value $0$ (black) stands for the foreground. The original image and the ground truth image are part of the dataset available\footnote{\url{http://www.ntu.edu.sg/home/asjfcai/Benchmark_Website/benchmark_index.html}.} in \cite{li13}. Our objective is to recover the original map $\sig$ from few user-inputs. To facilitate the interactions with the user, we divide the original image into the $\nbGroups = 600$ superpixels showed in Fig.~\ref{fig:image_example} and computed with SLIC~\cite{achanta12}, choose a small number of superpixels at random and ask the user to label these superpixels: $1$ if the superpixel belongs to the tiger; $0$ otherwise.

The graph $\Graph$ used to propagate the user-labels to the complete image is constructed as follows. We build a feature vector for each pixel by extracting a color RGB patch of size $3 \times 3$ around the pixel, transform this patch in vector form, and augment this vector with the absolute 2D coordinates of the pixels in this extracted patch. This yields $\nbVert$ feature vectors $\vec{g}_i \in \Rbb^{45}$, $i=1, \ldots, \nbVert$. We then connect each feature vector to its $9$ nearest neighbours (in the Euclidean sense), which gives a set of $9 \nbVert$ edges $\set{E}$. The adjacency matrix $\ma{W} \in \Rbb^{\nbVert \times \nbVert}$ satisfies
$
\ma{W}_{ij} := \exp [- {\|\vec{g}_i  - \vec{g}_j\|_2^2}/{\sigma^2}],
$
where $\sigma > 0$ is the $25^\th$ percentile of the set~$\{\|\vec{g}_i  - \vec{g}_j\|_2 : (i, j) \in \set{E} \}$. We finally symmetrise the matrix $\ma{W}$ and compute the combinatorial Laplacian $\Lap \in \Rbb^{\nbVert \times \nbVert}$.

We study three strategies to choose the superpixels. The first strategy consists in choosing the superpixels uniformly at random, \ie, using the sampling distribution $\vec{u}$. The second and third strategies consist in choosing the superpixels with respectively the optimised distributions $\bar{\vec{q}}$ and $\bar{\prob}$, which we evaluate at $\nbClass_{0} = 50$ using Jackson-Chebychev polynomials of order $75$ \cite{napoli13}. For illustration, we present in Fig.~\ref{fig:segmentation_result} the estimated values $\|\SelectGroup^{(\ell)}\Fou_{\nbClass_{0}}\|_F^2$ and $\|\SelectGroup^{(\ell)}\Fou_{\nbClass_{0}}\|_2^2$, which define the optimised distributions $\bar{\vec{q}}$ and $\bar{\vec{p}}$, respectively. Both distributions indicate that one should label more superpixels around the tiger. The distribution $\bar{\vec{q}}$ ``focuses'' however more on specific regions, like the head of tiger. The distribution $\bar{\prob}$ spreads the measurements over the entire tiger more uniformly.

We emulate user-interactions as follows. For each chosen superpixel, we compute the mean of the ground truth map $\sig$ within this superpixel. If the mean value is larger than $0.5$, we label the superpixel as part of the tiger. Otherwise, we label the superpixel as part of the background. This strategy obviously introduces noise if some superpixels cover part of the background and of the tiger. Once the labelling is done, we have access to the measurement vector $\tilde{\meas} \in \Rbb^\nbVertRed$ from which we want to reconstruct $\sig$. We repeat this procedure for $\nbGroupsRed \in \{50, 70, \ldots, 250\}$. For each $\nbGroupsRed$, we also repeat the experiments $50$ times with independent draws of the superpixels. We draw the superpixels \emph{with} replacements in all cases.

To reconstruct the original map $\sig$, we first use the fast reconstruction method~\eqref{eq:fast_decoder} and then refine the solution at the pixel level with~\eqref{eq:practical_decoder}, using the solution of the first minimisation problem as initialisation to solve the second minimisation problem. We choose $g(\Lap) = \Lap$ and solve both problems in the limit where $\reg \rightarrow 0$. In this limit, the problems \eqref{eq:fast_decoder} and \eqref{eq:practical_decoder} become
\begin{align}
\min_{\tilde{\vec{z}} \in \Rbb^{\nbGroups}} \; \tilde{\vec{z}}^\adjoint \, \widetilde{\Lap} \, \tilde{\vec{z}} 
\quad \text{ subject to } \quad 
\widetilde{\Meas} \tilde{\vec{z}} = \tilde{\meas}
\end{align}
and
\begin{align}
\min_{\vec{z} \in \Rbb^\nbVert} \; \vec{z}^\adjoint g(\Lap) \vec{z} 
\quad \text{ subject to } \quad 
\Meas \vec{z} = \meas,
\end{align}
respectively. Both problems are solved using FISTA~\cite{beck09a}. The same stopping criteria are used for all experiments.

\subsubsection{Results}

We present in the top panel of Fig.~\ref{fig:results_segmentation} the reconstruction snr obtained with the different methods. The reconstruction snr is defined as $-20\log_{10}(-\norm{\sig - \sig^*}_2/\norm{\sig}_2)$, where $\sig^*$ is the reconstructed signal.  We notice that the snr attained with the fast decoder~\eqref{eq:fast_decoder} is very similar to the snr attained with~\eqref{eq:practical_decoder}. We also remark that the optimised distributions $\bar{\prob}$ and $\bar{\vec{q}}$ yield better reconstructions than the uniform distribution $\vec{u}$. The mean reconstruction snr is slightly better with $\bar{\prob}$ than with $\bar{\vec{q}}$ at $s \geq 150$.

We present the computation time of each method in the bottom panel of Fig.~\ref{fig:results_segmentation}. We notice that solving~\eqref{eq:fast_decoder} is much faster than solving~\eqref{eq:practical_decoder}, while they yield almost the same quality. This highlight the interest of the fast reconstruction technique. It is also interesting to note that it is faster to solve~\eqref{eq:practical_decoder} when the measurements are drawn with $\bar{\prob}$ or with $\bar{\vec{q}}$ than with $\vec{u}$. The reason is probably a better initialisation of~\eqref{eq:practical_decoder} or a better ``quality'' of the measurements with the optimised distributions than with the uniform distribution.

Finally, we present in Fig.~\ref{fig:segmentation_result} some examples of reconstructions from $s=150$ sampled superpixels for each method. We notice that the optimised sampling distributions improve the reconstruction of $\sig$ around the head and tail of the tiger, \ie, where the optimised distributions have higher values. With a uniform distribution, the structure of the graph makes it difficult to reconstruct $\sig$ around the head and tail from the values of other superpixels. The optimised sampling distribution compensate this issue by favouring this area when selecting the measurements.

\begin{figure}
\centering
\includegraphics[width=.49\linewidth]{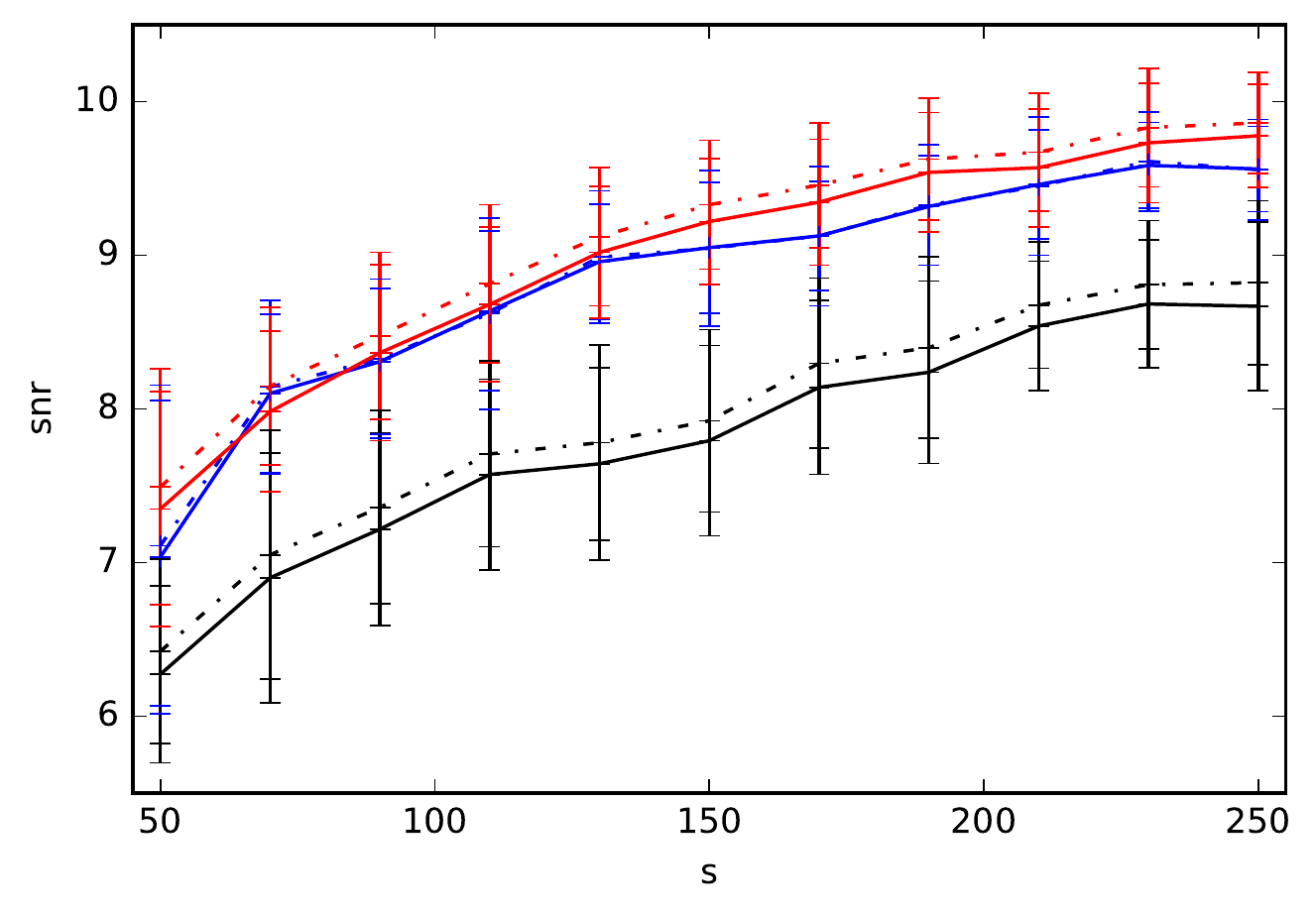}
\includegraphics[width=.49\linewidth]{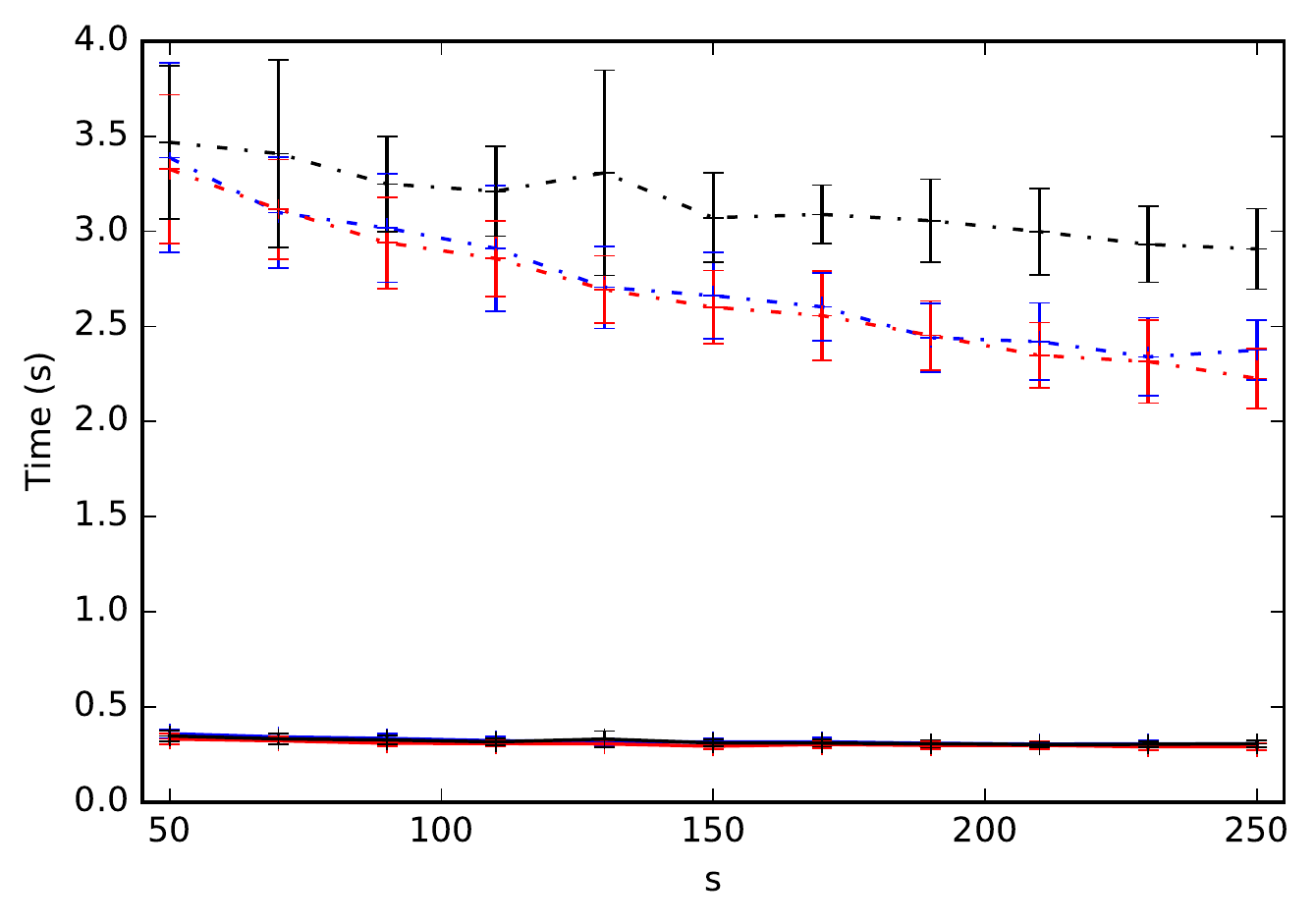}
\caption{\label{fig:results_segmentation}Left: the curves represent the mean reconstruction snr as a function of the number of sampled superpixels $\nbGroupsRed$. The dash-dotted curves are obtained by solving~\eqref{eq:practical_decoder}. The continuous curves are obtained by solving~\eqref{eq:fast_decoder}. Right: mean computation time in seconds as a function of $\nbGroupsRed$. The dash-dotted curves are obtained by solving~\eqref{eq:practical_decoder}. The continuous curves are obtained by solving~\eqref{eq:fast_decoder} using the result of~\eqref{eq:practical_decoder} as initial point. In both graphs, the errorbars are at one standard deviation. The black curves are obtained with the uniform sampling distribution $\vec{u}$. The blue curves are obtained with $\bar{\vec{q}}$. The red curves are obtained with $\bar{\prob}$.}
\end{figure}
%

%
\section{Discussion and Conclusion}

We presented a structured sampling strategy for $\nbClass$-bandlimited signals where the nodes are selected by groups. We proved that the local group graph cumulative coherence quantifies the importance of sampling each group to ensure a stable embedding of all $\nbClass$-bandlimited signals. Finally, we presented a fast reconstruction technique for $\nbClass$-bandlimited signals which are also nearly piecewise-constant over pre-defined groups of nodes.

Among the possible applications of these methods, we believe that they can also be useful to accelerate the compressive spectral clustering method proposed in~\cite{tremblay_ICML2016}. After having computed some features vectors for each nodes, this compressive method works by downsampling the set of features of vectors, performing $k$-means on this reduced set to find $\nbClass$ clusters, and interpolating the clustering results on all nodes by solving~\eqref{eq:practical_decoder}. To accelerate the method, one could 1) pre-group similar nodes to form $\nbGroups$ groups such that $\nbClass \leq \nbGroups \ll \nbVert$, \eg, by running few iterations of the $k$-means algorithm; 2) subsample this set of $\nbGroups$ groups; 3) cluster this subset to find $\nbClass$ clusters; and 4) solve \eqref{eq:fast_decoder} to cluster all nodes. If the overhead of computing the $\nbGroups$ groups is small, this method has the potential to be faster than the original compressive spectral clustering method.

Finally, we would like to discuss two limitations in the proposed methods. First, the optimal sampling distribution depends on the parameter $\nbClass$. In some applications, the final result may change a lot depending on the value of $\nbClass$ which was chosen to compute this distribution. Finding a range of values of $\nbClass$ which give acceptable and stable results is thus an important step in every application. Second, the estimation of the optimal sampling distribution depends on the quality of the polynomial approximation of the ideal low-pass filter $\eig_{\lambda_\nbClass}$. It is sometimes necessary to use a polynomial of large degree to get a correct estimation, which limits the computational efficiency of the proposed methods. In such cases, it would be especially useful to find more efficient alternatives to estimate the distributions $\prob^*$ and $\vec{q}^*$.

\ifregtemp
\appendix  
\titleformat{\section}[hang]{\color{blue1}\large\bfseries\centering}{Appendix \thesection}{0mm}{}[]
\else
\appendix
\fi

%
\ifregtemp
\section{ - Proof of the Theorem \ref{th:rip}}
\else
\section{Proof of the Theorem \ref{th:rip}}
\fi
\label{app:proof_rip}

As done in~\cite{puy15b}, the proof is obtained by applying the following lemma obtained by Tropp in \cite{tropp12}.

\begin{lemma}[Theorem $1.1$, \cite{tropp12}]
\label{th:matrix_chernoff}
Consider a finite sequence $\{ \ma{X}_j \}$ of independent, random, self-adjoint, positive semi-definite matrices of dimension $d \times d$. Assume that each random matrix satisfies
$
\lambda_{\rm max}(\ma{X}_j) \leq R
$
almost surely. Define
\begin{align}
\mu_{\rm min} := \lambda_{\rm min} \left( \sum_j \Ebb \, \ma{X}_j \right)
\; \text{and} \;
\mu_{\rm max} := \lambda_{\rm max} \left( \sum_j \Ebb \, \ma{X}_j \right).
\end{align}
Then
\begin{align}
\Pbb \left\{ \lambda_{\rm min} \left( \sum_j \ma{X}_j \right)  \leq (1 - \delta) \mu_{\rm min} \right\} 
& \leq d \, \left[ \frac{\ee^{-\delta}}{(1-\delta)^{1-\delta}}\right]^{\frac{\mu_{\rm min}}{R}}
\text{ for } \delta \in [0, 1],\\
\text{and } \;
\Pbb \left\{ \lambda_{\rm max} \left( \sum_j \ma{X}_j \right)  \geq (1 + \delta) \mu_{\rm max} \right\} 
& \leq d \, \left[ \frac{\ee^{\delta}}{(1+\delta)^{1+\delta}}\right]^{\frac{\mu_{\rm max}}{R}}
\text{ for } \delta \geq 0.
\end{align}
\end{lemma}

We will also use the facts that, for all $\delta \in [0, 1]$,
\begin{align}
\left[ \frac{\ee^{-\delta}}{(1-\delta)^{1-\delta}}\right]^{\mu_{\rm min}/R} 
\; \leq \; \exp \left( - \frac{\delta^2 \mu_{\rm min}}{3 \, R}\right)
\text{ and }
\left[ \frac{\ee^{\delta}}{(1+\delta)^{1+\delta}}\right]^{\mu_{\rm max}/R} 
\; \leq \; \exp \left( - \frac{\delta^2 \mu_{\rm max}}{3 \, R}\right).
\end{align}
\begin{proof}[Proof of Theorem \ref{th:rip}]
We start by noticing that
\begin{align}
\inv{\nbGroupsRed} \; \Fou_\nbClass^\adjoint \Meas^\adjoint \ma{P} \ma{P} \Meas \Fou_\nbClass = \inv{\nbGroupsRed} \sum_{j = 1}^\nbGroupsRed \left(\Fou_\nbClass^\adjoint {\SelectGroup^{(\omega_{j})}}^\adjoint\ma{P}^{(\omega_{j})} \right) \left(\ma{P}^{(\omega_{j})} \SelectGroup^{(\omega_l)} \Fou_\nbClass \right).
\end{align}
We define
\begin{align}
\ma{X}_{j} := \frac{1}{\nbGroupsRed} \left(\Fou_\nbClass^\adjoint {\SelectGroup^{(\omega_{j})}}^\adjoint\ma{P}^{(\omega_{j})} \right) \left(\ma{P}^{(\omega_{j})} \SelectGroup^{(\omega_{j})} \Fou_\nbClass \right)
\quad \text{ and } \quad
\ma{X} := \sum_{j=1}^\nbGroupsRed \ma{X}_{j} = \inv{\nbGroupsRed} \; \Fou_\nbClass^\adjoint \Meas^\adjoint \ma{P}^2 \Meas \Fou_\nbClass.
\end{align}
The matrix $\ma{X}$ is thus a sum of $\nbGroupsRed$ independent, random, self-adjoint, positive semi-definite matrices. We are in the setting of Lemma \ref{th:matrix_chernoff}. We continue by computing $\Ebb \,\ma{X}_{j}$ and $\lambda_{\rm max}(\ma{X}_j)$.

The expected value of each $\ma{X}_{j}$ is
\begin{align}
\label{eq:expected_value_proof}
\Ebb \, \ma{X}_{j} 
& = 
\Ebb \left[ \frac{1}{\nbGroupsRed} \left(\Fou_\nbClass^\adjoint {\SelectGroup^{(\omega_{j})}}^\adjoint\ma{P}^{(\omega_{j})} \right) \left(\ma{P}^{(\omega_{j})} \SelectGroup^{(\omega_{j})} \Fou_\nbClass \right) \right]
= 
\inv{\nbGroupsRed} \; \Fou_\nbClass^\adjoint \left(\sum_{\ell=1}^\nbGroups \prob_{\ell} \left({\SelectGroup^{(\ell)}}^\adjoint\ma{P}^{(\ell)} \right) \left(\ma{P}^{(\ell)} \SelectGroup^{(\ell)} \right)  \right) \Fou_\nbClass \nonumber \\
& =  
\inv{\nbGroupsRed} \; \Fou_\nbClass^\adjoint \left(\sum_{\ell=1}^\nbGroups {\SelectGroup^{(\ell)}}^\adjoint \SelectGroup^{(\ell)}  \right) \Fou_\nbClass
=
\inv{\nbGroupsRed} \; \Fou_\nbClass^\adjoint \Fou_\nbClass
=
\inv{\nbGroupsRed} \; \ma{I}.
\end{align}
Therefore, 
$
\lambda_{\rm min} \left( \sum_{j} \Ebb \, \ma{X}_{j} \right) = 1
$
and
$
\lambda_{\rm max} \left( \sum_{j} \Ebb \, \ma{X}_{j} \right) = 1.
$
Furthermore, for all $j = 1, \ldots, \nbGroupsRed$, we have
\begin{align}
\lambda_{\rm max} (\ma{X}_{j})
=
\norm{\ma{X}_{j}}_2 
\leq 
\max_{1 \leq \ell \leq \nbGroups} \norm{\frac{\ma{P}^{(\ell)} \SelectGroup^{(\ell)} \Fou_\nbClass}{\nbGroupsRed}}_2^2
=
\frac{1}{\nbGroupsRed} \; \max_{1 \leq \ell \leq \nbGroups} \left\{ \frac{\norm{\SelectGroup^{(\ell)} \Fou_\nbClass}_2^2}{\prob_{\ell}} \right\}
= \frac{\cumCoh_{\prob}^2}{\nbGroupsRed}.
\end{align}
Lemma \ref{th:matrix_chernoff} yields, for any $\delta \in (0, 1)$,
\begin{align}
\Pbb \left\{ \lambda_{\rm min} \left( \ma{X} \right)  \leq (1 - \delta) \right\} 
\; & \leq \; 
\nbClass \cdot \left[ \frac{\ee^{-\delta}}{(1-\delta)^{1-\delta}}\right]^{\nbGroupsRed/\cumCoh_{\prob}^2}
\; \leq \; 
\nbClass \; \exp \left( - \frac{\delta^2 \nbGroupsRed}{3 \, \cumCoh_{\prob}^2} \right) \\
\text{and } \;
\Pbb \left\{ \lambda_{\rm max} \left( \ma{X} \right)  \geq (1 + \delta) \right\} 
\; & \leq \; 
\nbClass \cdot \left[ \frac{\ee^{\delta}}{(1+\delta)^{1+\delta}}\right]^{\nbGroupsRed/\cumCoh_{\prob}^2}
\; \leq \; 
\nbClass \; \exp \left( - \frac{\delta^2 \nbGroupsRed}{3 \, \cumCoh_{\prob}^2} \right).
\end{align}
Therefore, for any $\delta \in (0, 1)$, we have, with probability at least $1 - \xi$,
\begin{align}
\label{eq:hidden_rip}
1 - \delta \leq \lambda_{\rm min} \left( \ma{X} \right) 
\quad \text{and} \quad
\lambda_{\rm max} \left( \ma{X} \right) \leq 1+\delta
\end{align}
provided that \eqref{eq:sampling_condition} holds. Noticing that \eqref{eq:hidden_rip} implies that
\begin{align}
(1-\delta) \norm{\vec{\alpha}}_2^2
\leq
\inv{\nbGroupsRed} \norm{\ma{P} \Meas \Fou_\nbClass \vec{\alpha} }_2^2 
\leq 
(1+\delta) \norm{\vec{\alpha}}_2^2,
\end{align}
for all $\vec{\alpha} \in \Rbb^k$, which is equivalent to \eqref{eq:RIP} for all $\sig \in \spann(\Fou_k)$, terminates the proof.
\end{proof}
%

%
\ifregtemp
\section{ - Proof of Theorem~\ref{th:fast_decoder}}
\else
\section{Proof of Theorem~\ref{th:fast_decoder}}
\fi
\label{app:proof_decoder}

In order to prove Theorem~\ref{th:fast_decoder}, we need to establish few properties between the different matrices used in this work.

The first useful property is
\begin{align}
\label{eq:commute_sample_average}
\widetilde{\ma{P}} \widetilde{\Meas} \Average = \widetilde{\Average} \ma{P} \Meas.
\end{align}
Indeed, for any $\vec{z} \in \Rbb^\nbVert$, the $j^\th$ entry of $\widetilde{\ma{P}} \widetilde{\Meas} \Average \vec{z}$ is
\begin{align}
\left( \widetilde{\ma{P}} \widetilde{\Meas} \Average \vec{z} \right)_{j}
=
\frac{\vec{1}^{\adjoint}\SelectGroup^{(\omega_{j})} \vec{z} }{\left( \prob_{\omega_{j}} \, \abs{\Group_{\omega_{j}}} \right)^{1/2}}.
\end{align}
Then, the $j^\th$ entry of $\widetilde{\Average} \ma{P} \Meas \vec{z}$ is the scaled sum of the values in the $j^\th$ sampled group appearing in $\ma{P} \Meas \vec{z}$, which is $\prob_{\omega_{j}}^{-1/2} \SelectGroup^{(\omega_{j})} \vec{z}$. From the definition of $\widetilde{\Average}$, the sum is scaled by $\abs{\Group_{\omega_{j}}}^{-1/2}$. Therefore, $( \widetilde{\ma{P}} \widetilde{\Meas} \Average \vec{z} )_{j} = ( \widetilde{\Average} \ma{P} \Meas \vec{z} )_{j}$ for all $j \in \{1, \ldots, \nbGroupsRed\}$, which terminates the proof.

The second property is 
\begin{align}
\label{eq:tight_frame_small_average}
\widetilde{\Average} \widetilde{\Average}^\adjoint = \ma{I},
\end{align}
which implies $\| \widetilde{\Average} \|_2 = 1$.

The third property is
\begin{align}
\label{eq:the_most_important_equality}
\| \widetilde{\Average} \ma{P} \Meas \Average^\adjoint \tilde{\vec{z}} \|_2 
= 
\norm{ \ma{P} \Meas \Average^\adjoint \tilde{\vec{z}} }_2,
\end{align}
for all $\tilde{\vec{z}} \in \Rbb^{\nbGroups}$. To prove this property, we remark that
$
\ma{P} \Meas \Average^\adjoint \tilde{\vec{z}}
= 
\widetilde{\Average}^\adjoint \widetilde{\ma{P}} \widetilde{\Meas} \tilde{\vec{z}}.
$
Indeed, the entries in $\ma{P} \Meas \Average^\adjoint \tilde{\vec{z}}$ corresponding to the first selected group $\Group_{\omega_1}$ are all equal to $(\prob_{\omega_1} \, \abs{\Group_{\omega_1}})^{-1/2} \tilde{\vec{z}}_{\omega_1}$. There are $\abs{\Group_{\omega_1}}$ such entries. One can also notice that the first $\abs{\Group_{\omega_1}}$ entries of $\widetilde{\Average}^\adjoint \widetilde{\ma{P}} \widetilde{\Meas} \tilde{\vec{z}}$ are also equal to $(\prob_{\omega_{j}} \, \abs{\Group_{\omega_1}})^{-1/2} \tilde{\vec{z}}_{\omega_1}$. Repeating this reasoning for all the sampled groups proves the equality. On the one side, we thus have
$
\norm{\ma{P} \Meas \Average^\adjoint \tilde{\vec{z}}}_2
= 
\| \widetilde{\Average}^\adjoint \widetilde{\ma{P}} \widetilde{\Meas} \tilde{\vec{z}} \|_2
=
\| \widetilde{\ma{P}} \widetilde{\Meas} \tilde{\vec{z}} \|_2,
$
where we used \eqref{eq:tight_frame_small_average}. On the other side, we have
$
\| \widetilde{\Average} \ma{P} \Meas \Average^\adjoint \tilde{\vec{z}} \|_2
= 
\| \widetilde{\ma{P}} \widetilde{\Meas} \Average \Average^\adjoint \tilde{\vec{z}} \|_2
=
\| \widetilde{\ma{P}} \widetilde{\Meas} \tilde{\vec{z}} \|_2,
$
where we used \eqref{eq:commute_sample_average} and \eqref{eq:tight_frame_average}. This terminates the proof.

\begin{proof}[Proof of Theorem \ref{th:fast_decoder}]
As $\sig^*$ is a minimiser of \eqref{eq:fast_decoder}, we have
\begin{align}
\label{eq:opt_cond_reg_problem}
\norm{\widetilde{\ma{P}}(\widetilde{\Meas} \tilde{\sig}^* - \tilde{\meas})}_2^2 + \reg \; (\tilde{\sig}^*)^\adjoint \, \widetilde{\Lap} \, \tilde{\sig}^*
\; \leq \;
\norm{\widetilde{\ma{P}}(\widetilde{\Meas} \tilde{\sig} - \tilde{\meas})}_2^2 + \reg \; \tilde{\sig}^\adjoint \, \widetilde{\Lap} \, \tilde{\sig}.
\end{align}
To prove the theorem, we need to lower and upper bound the left and right hand sides of \eqref{eq:opt_cond_reg_problem}, respectively. We start with the bound involving $\widetilde{\Lap}$ and then with the ones involving $\widetilde{\Meas}$.

\textbf{[Bounding the terms in \eqref{eq:opt_cond_reg_problem} involving $\widetilde{\Lap}$]}. We define the matrices
\begin{align}
\bar{\Fou}_\nbClass & := \left( \fou_{\nbClass+1}, \ldots, \fou_\nbVert \right) \in \Rbb^{\nbVert \times (\nbVert - \nbClass)},\\
\quad \ma{G}_\nbClass & := \diag\left( g(\eig_{1}), \ldots, g(\eig_{\nbClass}) \right) \in \Rbb^{\nbClass \times \nbClass},\\
%
\bar{\ma{G}}_\nbClass &:= \diag\left( g(\eig_{\nbClass+1}), \ldots, g(\eig_\nbVert) \right) \in \Rbb^{(\nbVert - \nbClass) \times (\nbVert - \nbClass)}.
\end{align}
By definition of $\vec{\alpha}^*$ and $\vec{\beta}^*$, $\Average^\adjoint \tilde{\sig}^* =  \vec{\alpha}^* + \vec{\beta}^*$ with $\vec{\alpha}^* \in \spann(\Fou_\nbClass)$ and $\vec{\beta}^* \in \spann(\bar{\Fou}_\nbClass)$. We recall that $\widetilde{\Lap} = (\Average\Fou) \; g(\Eig) \; (\Average\Fou)^\adjoint$. Therefore, we obtain 
\begin{align}
\label{eq:first_bound_L}
(\tilde{\sig}^*)^\adjoint \; \widetilde{\Lap} \; \tilde{\sig}^*
=
(\Fou_\nbClass^\adjoint \vec{\alpha}^*)^\adjoint \; \ma{G}_\nbClass \; (\Fou_\nbClass^\adjoint \vec{\alpha}^*)
+ (\bar{\Fou}_\nbClass^\adjoint \vec{\beta}^*)^\adjoint \; \bar{\ma{G}}_\nbClass \; (\bar{\Fou}_\nbClass^\adjoint \vec{\beta}^*)
\geq
g(\eig_{\nbClass+1}) \norm{\vec{\beta}^*}_2^2.
\end{align}
In the first step, we used the facts that $\bar{\Fou}_\nbClass^\adjoint \vec{\alpha}^* = \vec{0}$ and $\Fou_\nbClass^\adjoint \vec{\beta}^* = \vec{0}$. The second step follows form the fact that $\norm{ \bar{\Fou}_\nbClass^\adjoint \vec{\beta}^* }_2 = \norm{ \vec{\beta}^* }_2$. We also have
\begin{align}
\label{eq:second_bound_L}
\tilde{\sig}^\adjoint \, \widetilde{\Lap} \, \tilde{\sig}
& =
(\Fou^\adjoint \Average^\adjoint \tilde{\sig})^\adjoint \; g(\Eig) \; (\Fou^\adjoint \Average^\adjoint \tilde{\sig})
\; \leq \; 
g(\eig_\nbClass) \norm{\Fou_\nbClass \Average^\adjoint \tilde{\sig}}_2^2
+ 
g(\eig_\nbVert) \norm{\bar{\Fou}_\nbClass^\adjoint \Average^\adjoint \tilde{\sig}}_2^2 \nonumber \\
& \leq
g(\eig_\nbClass) \norm{\Average^\adjoint \Average \sig}_2^2
+ 
g(\eig_\nbVert) \norm{\bar{\Fou}_\nbClass^\adjoint \Average^\adjoint \Average \sig}_2^2
\; \leq \;
g(\eig_\nbClass) \norm{\sig}_2^2
+ 
g(\eig_\nbVert) \left[ \norm{\bar{\Fou}_\nbClass^\adjoint \sig}_2^2
+ 
\norm{\bar{\Fou}_\nbClass^\adjoint (\Average^\adjoint \Average \sig - \sig)}_2^2 \right] \nonumber \\
& \leq
g(\eig_\nbClass) \norm{\sig}_2^2
+ 
\epsilon^2 g(\eig_\nbVert) \norm{\sig}_2^2.
\end{align}
The second inequality follows from the facts that $\norm{\Fou_k}_2 = 1$ and $\tilde{\sig} = \Average \sig$. To obtain the third inequality, we used $\norm{\Average}_2 = 1$ and the triangle inequality. For the last step, notice that $\norm{\bar{\Fou}_\nbClass^\adjoint \sig}_2 = 0$ (as $\sig \in \spann(\Fou_{\nbClass})$), $\norm{\bar{\Fou}_k}_2 = 1$ and use~\eqref{eq:blockconstant_signal}. 

\textbf{[Bounding the terms in \eqref{eq:opt_cond_reg_problem} involving $\widetilde{\Meas}$]}. By definition of $\tilde{\meas}$, it is immediate that
\begin{align}
\label{eq:first_bound_M}
\norm{\widetilde{\ma{P}}(\widetilde{\Meas} \tilde{\sig} - \tilde{\meas})}_2^2
=
\norm{\widetilde{\ma{P}}\tilde{\err}}_2^2.
\end{align}
For the other term involving $\widetilde{\Meas}$,
the triangle inequality yields
\begin{align*}
\norm{\widetilde{\ma{P}}\widetilde{\Meas} \tilde{\sig}^* - \widetilde{\ma{P}}\tilde{\meas}}_2 
& \geq
\norm{\widetilde{\ma{P}}\widetilde{\Meas} \tilde{\sig}^* - \widetilde{\ma{P}}\widetilde{\Meas} \Average\sig}_2 - \norm{ \widetilde{\ma{P}} \tilde{\err} }_2.
\end{align*}
Then, we have
\begin{align}
\norm{\widetilde{\ma{P}}\widetilde{\Meas} \tilde{\sig}^* - \widetilde{\ma{P}}\widetilde{\Meas} \Average\sig}_2
& =
\norm{\widetilde{\ma{P}}\widetilde{\Meas} \Average  \Average^\adjoint \tilde{\sig}^* - \widetilde{\ma{P}}\widetilde{\Meas} \Average\sig}_2
= 
\norm{\widetilde{\Average} \ma{P} \Meas  \Average^\adjoint \tilde{\sig}^* - \widetilde{\Average} \ma{P} \Meas \sig}_2 \\
& \geq
\norm{\tilde{\Average} \ma{P} \Meas  \Average^\adjoint \tilde{\sig}^* - \tilde{\Average} \ma{P} \Meas \Average^\adjoint \Average \sig}_2
- \norm{\tilde{\Average} \ma{P} \Meas \Average^\adjoint \Average \sig - \tilde{\Average} \ma{P} \Meas \sig}_2.
\end{align}
The first equality follows from $\Average \Average^\adjoint = \ma{I}$, the second from \eqref{eq:commute_sample_average}, and the triangle inequality was used in the last step. To summarise, we are at
\begin{align}
\label{eq:lower_bound_fidelity_term}
\norm{\widetilde{\ma{P}}\widetilde{\Meas} \tilde{\sig}^* - \widetilde{\ma{P}}\tilde{\meas}}_2
\geq
\norm{\tilde{\Average} \ma{P} \Meas  \Average^\adjoint \tilde{\sig}^* - \tilde{\Average} \ma{P} \Meas \Average^\adjoint \Average \sig}_2
- \norm{\tilde{\Average} \ma{P} \Meas \Average^\adjoint \Average \sig - \tilde{\Average} \ma{P} \Meas \sig}_2 - \norm{ \widetilde{\ma{P}} \tilde{\err} }_2.
\end{align}
We continue by lower bounding $\| \tilde{\Average} \ma{P} \Meas  \Average^\adjoint \tilde{\sig}^* - \tilde{\Average} \ma{P} \Meas \Average^\adjoint \Average \sig \|_2$ and upper bounding $\| \tilde{\Average} \ma{P} \Meas \Average^\adjoint \Average \sig - \tilde{\Average} \ma{P} \Meas \sig \|_2$ separately.

\emph{[Lower bound on $\| \tilde{\Average} \ma{P} \Meas  \Average^\adjoint \tilde{\sig}^* - \tilde{\Average} \ma{P} \Meas \Average^\adjoint \Average \sig \|_2$]}. Equality \eqref{eq:the_most_important_equality} yields
\begin{align}
\norm{\tilde{\Average} \ma{P} \Meas  \Average^\adjoint \tilde{\sig}^* - \tilde{\Average} \ma{P} \Meas \Average^\adjoint \Average \sig}_2
= 
\norm{\ma{P} \Meas  \Average^\adjoint \tilde{\sig}^* - \ma{P} \Meas \Average^\adjoint \Average \sig}_2.
\end{align}
Using the triangle inequality and the fact that $\sig \in \PSignals$, we obtain
\begin{align}
\norm{\ma{P} \Meas  \Average^\adjoint \tilde{\sig}^* - \ma{P} \Meas \Average^\adjoint \Average \sig}_2
& \geq \norm{\ma{P} \Meas  \Average^\adjoint \tilde{\sig}^* - \ma{P} \Meas \sig}_2 - \norm{\ma{P} \Meas \sig - \ma{P} \Meas \Average^\adjoint \Average \sig}_2 \nonumber \\
& \geq \norm{\ma{P} \Meas  \Average^\adjoint \tilde{\sig}^* - \ma{P} \Meas \sig}_2 - \norm{\ma{P} \Meas}_2 \norm{\Average^\adjoint \Average \sig - \sig}_2 \nonumber \\
& \geq \norm{\ma{P} \Meas  \Average^\adjoint \tilde{\sig}^* - \ma{P} \Meas \sig}_2 - \epsilon \, M_{\rm max} \norm{ \sig }_2.
\end{align}
The restricted isometry property~\eqref{eq:RIP} then yields
\begin{align}
\norm{\ma{P} \Meas  \Average^\adjoint \tilde{\sig}^* - \ma{P} \Meas \sig}_2
& =
\norm{ \ma{P} \Meas \, (\vec{\alpha}^* - \sig) + \ma{P} \Meas \vec{\beta}^*}_2
\geq
\norm{\ma{P} \Meas (\vec{\alpha}^* - \sig)}_2 - \norm{\ma{P}\Meas \vec{\beta}^*}_2 \nonumber \\
& \geq
\sqrt{\nbGroupsRed (1-\delta)} \; \norm{ \vec{\alpha}^* - \sig }_2 - M_{\rm max} \norm{\vec{\beta}^*}_2.
\end{align}
We used the equality $\Average^\adjoint \tilde{\sig}^* = \vec{\alpha}^* + \vec{\beta}^*$. We thus proved that
\begin{align}
\label{eq:lower_bound_with_rip}
\| \tilde{\Average} \ma{P} \Meas  \Average^\adjoint \tilde{\sig}^* - \tilde{\Average} \ma{P} \Meas \Average^\adjoint \Average \sig \|_2
\geq
\sqrt{\nbGroupsRed (1-\delta)} \; \norm{ \vec{\alpha}^* - \sig }_2 - M_{\rm max} \norm{\vec{\beta}^*}_2 - \epsilon M_{\rm max} \norm{ \sig }_2.
\end{align}

\emph{[Upper bounding $\| \tilde{\Average} \ma{P} \Meas \Average^\adjoint \Average \sig - \tilde{\Average} \ma{P} \Meas \sig \|_2$]}. We have
\begin{align}
\label{eq:lower_bound_with_piecewise_constraint}
\norm{\tilde{\Average} \ma{P} \Meas \Average^\adjoint \Average \sig - \tilde{\Average} \ma{P} \Meas \sig}_2
\leq
\| \tilde{\Average} \|_2
\norm{\ma{P} \Meas}_2
\norm{\Average^\adjoint \Average \sig - \sig }_2
\leq \epsilon \, M_{\rm max} \norm{ \sig }_2.
\end{align}
We used the facts that $\| \tilde{\Average} \|_2 = 1$ (see \eqref{eq:tight_frame_small_average}) and that $\sig \in \PSignals$ in the last inequality.

Using the inequalities \eqref{eq:lower_bound_with_rip} and \eqref{eq:lower_bound_with_piecewise_constraint} in \eqref{eq:lower_bound_fidelity_term}, we arrive at
\begin{align}
\label{eq:second_bound_M}
\norm{\widetilde{\ma{P}}\widetilde{\Meas} \tilde{\sig}^* - \widetilde{\ma{P}}\tilde{\meas}}_2
\geq
\sqrt{\nbGroupsRed (1-\delta)} \; \norm{ \vec{\alpha}^* - \sig }_2 - M_{\rm max} \norm{\vec{\beta}^*}_2
- 2\epsilon M_{\rm max} \norm{ \sig }_2 - \norm{ \widetilde{\ma{P}} \tilde{\err} }_2.
\end{align}

\textbf{[Finishing the proof]}. Remark that \eqref{eq:opt_cond_reg_problem} implies
\begin{align}
\label{eq:opt_cond_1}
\norm{\widetilde{\ma{P}}(\widetilde{\Meas} \tilde{\sig}^* - \tilde{\meas})}_2^2 
\leq
\norm{\widetilde{\ma{P}}(\widetilde{\Meas} \tilde{\sig} - \tilde{\meas})}_2^2 + \reg \; \tilde{\sig}^\adjoint \, \widetilde{\Lap} \, \tilde{\sig},\\
\text{and} \quad
\label{eq:opt_cond_2}
\reg \; (\tilde{\sig}^*)^\adjoint \, \widetilde{\Lap} \, \tilde{\sig}^*
\leq
\norm{\widetilde{\ma{P}}(\widetilde{\Meas} \tilde{\sig} - \tilde{\meas})}_2^2 + \reg \; \tilde{\sig}^\adjoint \, \widetilde{\Lap} \, \tilde{\sig}.
\end{align}
Using \eqref{eq:first_bound_L}, \eqref{eq:second_bound_L}, \eqref{eq:first_bound_M} in \eqref{eq:opt_cond_2}, we obtain
\begin{align}
\reg \; g(\eig_{\nbClass+1}) \norm{\vec{\beta}^*}_2^2
\leq \;
\norm{ \widetilde{\ma{P}} \err }_2^2 
+ \; \reg \; g(\eig_\nbClass) \norm{\sig}_2^2
+ \epsilon^2 \; \reg \; g(\eig_\nbVert) \norm{\sig}_2^2,
\end{align}
which implies \eqref{eq:bound_beta} in Theorem~\ref{th:fast_decoder}. It remains to prove \eqref{eq:bound_alpha} to finish the proof.

Using \eqref{eq:second_bound_L}, \eqref{eq:first_bound_M} and \eqref{eq:second_bound_M} in \eqref{eq:opt_cond_1}, we obtain
\begin{align}
\sqrt{\nbGroupsRed (1-\delta)} \norm{\vec{\alpha}^* - \sig }_2
\leq
2 \norm{ \widetilde{\ma{P}} \err }_2 + M_{\rm max} \norm{\vec{\beta}^*}_2
+ (\sqrt{\reg g(\eig_\nbClass)} + \epsilon \sqrt{\reg g(\eig_\nbVert)} + 2 \epsilon M_{\rm max}) \norm{ \sig }_2.
\end{align}
Using \eqref{eq:bound_beta} to bound $\norm{\vec{\beta}^*}_2$ on the right hand side, we have
\begin{align}
\sqrt{\nbGroupsRed (1-\delta)} \norm{\vec{\alpha}^* - \sig }_2
\; & \leq \; 
2 \norm{ \widetilde{\ma{P}} \err }_2 + \frac{M_{\rm max}}{\sqrt{\reg g(\eig_{\nbClass+1})}} \norm{\widetilde{\ma{P}} \err}_2
+ M_{\rm max} \sqrt{\frac{g(\eig_\nbClass)}{g(\eig_{\nbClass+1})}} \norm{ \sig  }_2 \nonumber \\
& + \epsilon \;  M_{\rm max} \sqrt{\frac{g(\eig_\nbVert)}{g(\eig_{\nbClass+1})}} \norm{ \sig  }_2
+ \left(\sqrt{\reg g(\eig_\nbClass)} + \epsilon \sqrt{\reg g(\eig_\nbVert)} + 2 \epsilon M_{\rm max} \right) \norm{ \sig }_2  \\
& =
\left( 2 +  \frac{M_{\rm max}}{\sqrt{\reg g(\eig_{\nbClass+1})}}\right)\norm{\widetilde{\ma{P}} \err}_2
+ \left( M_{\rm max} \sqrt{\frac{g(\eig_\nbClass)}{g(\eig_{\nbClass+1})}} + \sqrt{\reg g(\eig_\nbClass)}\right) \norm{ \sig  }_2 \nonumber \\
& + \epsilon \left( 2 M_{\rm max} + M_{\rm max} \sqrt{\frac{g(\eig_\nbVert)}{g(\eig_{\nbClass+1})}} + \sqrt{\reg g(\eig_\nbVert)} \right) \norm{ \sig  }_2. 
\end{align}
We only rearranged the term in the last step. This proves \eqref{eq:bound_alpha} and terminates the proof.
\end{proof}
%

\ifregtemp
	\bibliographystyle{IEEEtran}
\else
	\bibliographystyle{imaiai}
\fi
\bibliography{biblio,imedit}

\end{document}